\documentclass[a4paper]{quantumarticle}
\pdfoutput=1

\usepackage[utf8]{inputenc}
\usepackage[english]{babel}
\usepackage[T1]{fontenc}
\usepackage{amsmath,amssymb,amsfonts,amsthm}

\newtheorem{theorem}{Theorem}[section]
\newtheorem{corollary}[theorem]{Corollary}
\newtheorem{definition}[theorem]{Definition}
\newtheorem{proposition}[theorem]{Proposition}
\newtheorem{assumption}[theorem]{Assumption}
\usepackage{graphicx}
\usepackage{float}
\usepackage[section]{placeins}
\usepackage{hyperref}
\usepackage{booktabs}
\usepackage{bm}
\usepackage{mathtools}

\newcommand{\minimize}{\operatorname*{minimise}}
\newcommand{\ZC}{\mathbb{Z}_C}
\newcommand{\ZCn}{\mathbb{Z}_C^n}
\newcommand{\DC}{D_C}
\newcommand{\Sk}{S_k}
\newcommand{\dmax}{d_{\max}}

\begin{document}

\title{Quantum Speedup for Network Coordination via Fourier Sparsity}

\author{Vinayak Dixit}
\affiliation{School of Civil and Environmental Engineering, University of New South Wales, Sydney, Australia}

\date{March 8, 2026}

\maketitle

\begin{abstract}
Network coordination---synchronising traffic signals, scheduling trains, assigning communication slots---requires minimising pairwise costs across coupled systems. These problems are NP-hard yet share a common Fourier-sparse structure exploitable by quantum algorithms. We introduce the Fourier Network Coordination problem (Fourier-NC), unifying eight application domains. For abelian and dihedral groups, classical sparse Fourier transforms match quantum in the same oracle model, limiting the advantage to at most polynomial. The genuine separation emerges for the symmetric group~$\Sk$: a conditional super-exponential speedup of $k!\to\mathrm{poly}(k)$ for class-function costs with non-trivial minimisers. When the minimising conjugacy class is structurally determined, the problem lies in NP$\,\cap\,$BQP and is conditionally outside~P (Corollary~\ref{cor:bqp}), placing it in the intermediate complexity regime alongside integer factorisation and graph isomorphism. We formalise the \emph{abelian index} $\alpha(G)=[G:A_{\max}]$ as the structural invariant governing the quantum--classical gap and identify a three-regime complexity trichotomy: abelian ($\alpha=1$, classical sFFT suffices), nearly abelian ($\alpha=\dmax$, polynomial advantage), and strongly non-abelian ($\alpha\gg \dmax$, super-exponential advantage).
\end{abstract}

\section{Introduction}
\label{sec:intro}

Many real-world systems require coordinating timing or ordering across a network: synchronising traffic lights so vehicles hit green waves, scheduling trains to minimise passenger waiting, assigning communication slots to avoid interference, or sequencing signal phases to maintain coordinated vehicle flow. These problems share a common mathematical structure---each pair of neighbours pays a cost that depends only on their \emph{relative offset or ordering}---and that structure is precisely what a quantum computer can exploit. The general problem is NP-hard (we prove this via reduction from MAX-CUT, Theorem~\ref{thm:nphard}), but its Fourier-sparse structure enables exponential quantum speedup for an important subclass---specifically, frustration-free instances whose cost functions have polynomially bounded Fourier dynamic range (a condition verified for all engineering domains considered here; see Propositions~\ref{prop:pmin_cosine}--\ref{prop:pmin_pwl}).

\subsection{Intuitive overview}
\label{sec:intuition}

Consider four traffic lights along a corridor, each cycling through a 60-second period. Each must choose a start-time offset; if adjacent lights differ by 15~seconds, vehicles hit a ``green wave.'' With $60^{100}$ possible assignments for a 100-intersection network, brute-force search is hopeless---but the delay function between each pair of lights is \emph{smooth and periodic}, describable by two or three Fourier components. The Quantum Fourier Transform identifies these components in one sweep, collapsing the search from exponential to $O(mr)$ modes. For cyclic offsets ($\ZC$), classical sparse Fourier transforms match this advantage; the genuine quantum speedup emerges for \emph{orderings} over the symmetric group~$\Sk$, where $k!$~possibilities and no known classical sFFT yield a conditional $k!\to\mathrm{poly}(k)$ separation. The algorithm requires frustration-free networks (Definition~\ref{def:frustration}); frustrated instances receive bounded approximations (Proposition~\ref{prop:frustration_gap}).

\subsection{Contributions}

\begin{enumerate}
\item We define Fourier-NC over~$\ZC$, prove the Fourier Factorization Theorem (Theorem~\ref{thm:fourier_fact}), and give a polynomial-time quantum algorithm (Theorem~\ref{thm:correctness}) with NP-completeness (Theorem~\ref{thm:nphard}) and an $\Omega(C^n)$ classical lower bound (Theorem~\ref{thm:lower_bound}).
\item We formalise the \emph{abelian index} $\alpha(G)=[G:A_{\max}]$ (Definition~\ref{def:abelian_index}), prove the fundamental inequality $\dmax(G)\le\alpha(G)$ (Proposition~\ref{prop:fundamental_ineq}), and establish a three-regime complexity trichotomy. Classical sFFT~\cite{Hassanieh2012} matches quantum for $r$-sparse abelian instances (Sec.~\ref{sec:sfft_gap}).
\item We extend to $\Sk$ (Permutation Coordination) and establish a conditional super-exponential speedup of $k!\to\mathrm{poly}(k)$ (Proposition~\ref{thm:sk_separation}, conditional on Assumption~\ref{ass:no_sfft}) for class-function costs. When the minimising conjugacy class is structurally determined, the problem is in BQP and conditionally not in~P (Corollary~\ref{cor:bqp}), placing it in the intermediate regime alongside factoring and graph isomorphism.
\item We extend to directed networks via~$\DC$ with $4\times$ overhead (Appendix~\ref{app:dihedral}) and provide gate count analysis (Table~\ref{tab:gates_projected}).
\end{enumerate}

\noindent Sections~\ref{sec:framework}--\ref{sec:complexity} develop the abelian ($\ZC$) framework; Section~\ref{sec:symmetric} extends to~$\Sk$, where the genuine super-exponential advantage arises.

\section{Literature Review and Related Work}
\label{sec:litreview}

The coordination problems we unify share a common algebraic structure: minimisation of pairwise costs over cyclic offsets, with Fourier-sparse cost functions (Fig.~\ref{fig:fourier_spectrum}). Despite surface-level differences, eight domains share this structure.

\begin{figure}[!htbp]
\centering
\includegraphics[width=\columnwidth]{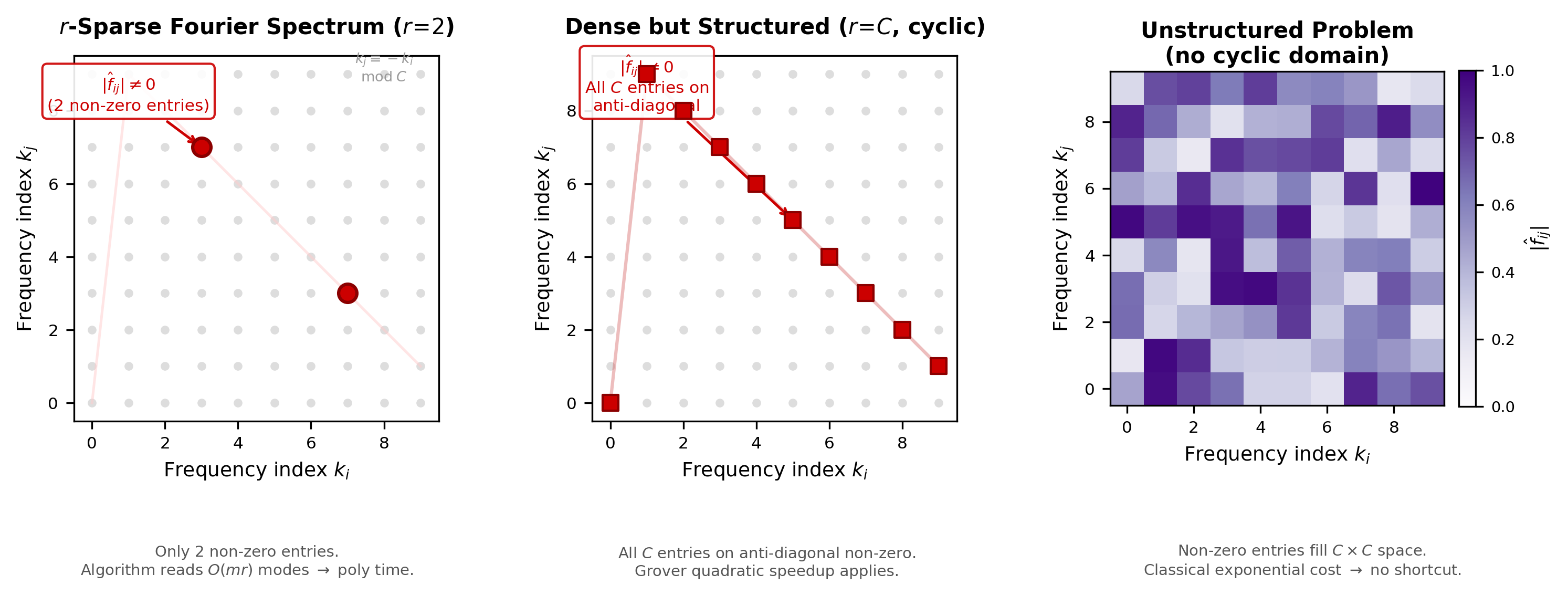}
\caption{Fourier spectrum $|\hat H(k)|$ for a single edge ($C=10$). Left: 2-sparse ($r=2$)---only 2 non-zero entries on the anti-diagonal $k_j=-k_i\bmod C$. Centre: dense but cyclic ($r=C$) with $O(C)$ modes, Grover quadratic speedup. Right: unstructured---$O(C^2)$ modes, no quantum shortcut. The Fourier-NC algorithm exploits left-panel structure.}
\label{fig:fourier_spectrum}
\end{figure}

\subsection{Application domains}

Eight coordination domains share the Fourier-NC structure (Table~\ref{tab:domains}): traffic signal coordination~\cite{Gartner1975,GartnerDeshpande2009,DixitPech2026,Little1981,Gartner1991}, railway timetabling (PESP)~\cite{Serafini1989,LindnerLiebchen2022}, wireless communications (TDMA/FAP)~\cite{RamaswamiParhi1989,Gronkvist2005,Hale1980,Aardal2007}, power systems (OPF)~\cite{BienstockVerma2019,LavaeiLow2012}, oscillator synchronisation~\cite{Kuramoto1975,Strogatz2000}, clock synchronisation~\cite{Awerbuch1985}, and influence scheduling~\cite{OlfatiSaber2007,Kempe2003}. In every case, smooth periodic coupling produces cost functions well-approximated by a short Fourier series ($r=1$--$5$). NP-hardness has been established independently in each domain~\cite{Dauscha1985,Serafini1989,RamaswamiParhi1989,Hale1980,BienstockVerma2019}.

\begin{table}[!htbp]
\centering
\caption{Complexity landscape across application domains. All domains are abelian ($\ZC$); the quantum advantage is an oracle-model separation (at most polynomial for $r$-sparse instances; see Sec.~\ref{sec:sfft_gap}).}
\label{tab:domains}
\resizebox{\columnwidth}{!}{%
\begin{tabular}{@{}lllcl@{}}
\toprule
Domain & Classical & Best classical & $r$ & Fourier-NC$^\dagger$ \\
\midrule
Traffic NC & NP-hard & heuristics & 1--3 & Poly (oracle sep.) \\
Railway PESP & NP-complete & IP solvers & 2--5 & Poly (oracle sep.) \\
TDMA slot & NP-hard & heuristics & 2--4 & Poly (oracle sep.) \\
FAP & NP-hard & greedy & 2--4 & Poly (oracle sep.) \\
Clock sync & NP-hard & consensus & 1--2 & Poly (oracle sep.) \\
OPF & NP-hard & SDP relax & 1 & Poly (oracle sep.) \\
Kuramoto & NP-hard & gradient flow & 1 & Poly (oracle sep.) \\
Influence sched. & NP-hard & iterative & 1--2 & Poly (oracle sep.) \\
\bottomrule
\multicolumn{5}{@{}l}{\footnotesize $^\dagger$Bipartite/frustration-free instances; Sec.~\ref{sec:complexity_dichotomy}.}
\end{tabular}}
\end{table}

\paragraph{Quantum approaches to combinatorial optimisation.}
Existing quantum paradigms---Grover search~\cite{Grover1996,Grover1997}, QAOA~\cite{Farhi2014}, quantum annealing~\cite{AlbashLidar2018}, and quantum walk algorithms~\cite{Childs2003}---treat the cost function as a black box, achieving at most quadratic speedup for generic NP-hard instances. Quantum walks exploit graph topology for search and sampling but do not leverage the \emph{algebraic} Fourier sparsity of cost functions, which is the structural feature Fourier-NC exploits. The QAOA cost operator $e^{-i\gamma H}$ is structurally identical to our phase oracle, but QAOA provides only approximate solutions via variational optimisation, whereas Fourier-NC achieves exact recovery in polynomial time by exploiting Fourier sparsity.

\paragraph{The Fourier structure gap.}
Real-world coordination costs have low Fourier sparsity ($r=1$--$5$), making the $C^n$-dimensional landscape effectively low-rank with $O(mr)$ degrees of freedom. The QFT collapses the search from $C^n$ to $O(mr)$ modes---the first network coordination problem where Fourier sparsity (rather than number-theoretic periodicity or a black-box promise) yields a quantum oracle-model separation analogous to HHL~\cite{HHL2009}. For $r$-sparse abelian instances, classical sFFT closes the gap (Sec.~\ref{sec:sfft_gap}); the $\Sk$ extension recovers a genuine conditional speedup.

\section{The Fourier-NC Framework}
\label{sec:framework}

The central idea is simple: each node in a network must choose a setting (a time slot, a phase angle, a channel) from a finite set of $C$~options arranged in a cycle. Neighbouring nodes pay a cost that depends only on the \emph{difference} between their settings. The question is: what assignment minimises the total cost?

\paragraph{Notation.}
Let $G=(V,\mathcal{A})$ denote a network with $n=|V|$ nodes and $m=|\mathcal{A}|$ edges. $\ZC=\mathbb{Z}/C\mathbb{Z}$ is the cyclic group of order~$C$; $\DC$ the dihedral group of order~$2C$; $\Sk$ the symmetric group on $k$~elements. $\hat{f}(k)$ denotes the Fourier coefficient of~$f$ at frequency~$k$, and $r$ is the Fourier sparsity (number of non-zero coefficients). For~$\Sk$, $\lambda\vdash k$ indexes irreducible representations with dimension~$d_\lambda$; $p(k)$ is the number of partitions of~$k$; $\chi^\lambda$ is the character of irrep~$\lambda$.

For each edge $(i,j)\in\mathcal{A}$, a cost function $f_{ij}\colon\ZC\to\mathbb{R}$ maps offset differences to costs. The global cost is $H(\mu)=\sum_{(i,j)\in\mathcal{A}}f_{ij}\!\bigl((\mu_i-\mu_j)\bmod C\bigr)$. The \textbf{Fourier-NC problem asks}: $\minimize H(\mu)$ over~$\ZCn$.

\textbf{Decision version (Fourier-NC Decision).} Given a threshold $\tau$, decide whether there exists $\mu\in\ZCn$ with $H(\mu)\le\tau$. Since a candidate $\mu$ can be verified in polynomial time, Fourier-NC Decision is in NP; combined with Theorem~\ref{thm:nphard}, it is NP-complete.

The key observation is that in all eight application domains (Sec.~\ref{sec:litreview}), the cost functions are \emph{smooth and periodic}, so their Fourier transforms are \emph{sparse}---only a handful of frequency components matter. This is what makes the problem tractable for a quantum computer.

\medskip

\begin{definition}[$r$-Sparse Delay Function]
A cost function $f_{ij}$ is $r$-Fourier-sparse if its DFT has at most $r$ non-zero coefficients, i.e.\
\begin{equation}
\hat f_{ij}(k) = \frac{1}{C}\sum_{x=0}^{C-1} f_{ij}(x)\cdot\omega_C^{-kx},
\label{eq:dft_edge}
\end{equation}
where $\omega_C=e^{2\pi i/C}$, and $|\{k:\hat f_{ij}(k)\neq0\}|\le r$.
\end{definition}

Three conditions characterise Fourier-NC applicability: \textbf{(C1)}~\emph{Group domain}---variables in $G^n$ ($G=\ZC,\DC,\Sk$); \textbf{(C2)}~\emph{Pairwise relative structure}---cost $H=\sum f_{ij}(g_i^{-1}g_j)$, so the transform factorises edge-by-edge; \textbf{(C3)}~\emph{Fourier sparsity}---each $f_{ij}$ has $r\ll|G|$ non-zero coefficients. All eight domains in Sec.~\ref{sec:litreview} satisfy (C1)--(C3); when any condition fails, only generic algorithms (Grover, QAOA) apply.

\subsection{Fourier Factorization}

\begin{theorem}[Fourier Factorization]
\label{thm:fourier_fact}
If each $f_{ij}$ is $r$-sparse, then the DFT of $H$ over $\ZCn$ has at most $O(mr)$ non-zero coefficients.
\end{theorem}

\begin{proof}[Proof sketch (full proof in Appendix~\ref{app:fourier_fact_proof})]
By character orthogonality on~$\ZC$, the single-edge DFT $\hat H_{ij}(\mathbf{k})$ vanishes unless $k_\ell=0$ for all $\ell\neq i,j$ and $k_j\equiv -k_i\pmod{C}$, reducing to $\hat H_{ij}(\mathbf{k})=\hat f_{ij}(k_i)$. Since each $f_{ij}$ has at most $r$ non-zero coefficients, each edge contributes at most $r$ non-zero frequency vectors. Summing over $m$ edges gives at most $mr$ potentially non-zero modes.
\end{proof}

\subsection{Frustration-freeness}

\begin{definition}[Frustration-free and frustrated graph]\label{def:frustration}
For each edge, let $D^*_{ij}=\arg\min_d f_{ij}(d)$. The graph~$G$ is \emph{frustration-free} if there exists a selection $d^*_{ij}\in D^*_{ij}$ for every edge and an assignment $\mu\in\ZCn$ with $\mu_i-\mu_j\equiv d^*_{ij}\pmod{C}$ for every edge; equivalently, the \emph{holonomy} $h_\gamma=\sum_{(i,j)\in\gamma}d^*_{ij}\bmod C$ vanishes for every cycle~$\gamma$. Otherwise $G$ is \emph{frustrated}.
\end{definition}

For frustration-free graphs the tree solver (Sec.~\ref{sec:algorithm}) is exact; for frustrated graphs the suboptimality gap is bounded by the cycle rank $\beta=m-n+1$ (Proposition~\ref{prop:frustration_gap}, Appendix~\ref{app:frustration}). Bipartite graphs are always frustration-free (holonomy vanishes on even cycles).

\section{The Quantum Algorithm}
\label{sec:algorithm}

The algorithm encodes the cost function into quantum phases, then applies the QFT to extract the $O(mr)$ dominant frequencies (Theorem~\ref{thm:fourier_fact}). We distinguish \emph{query complexity} (oracle calls) from \emph{gate complexity} (elementary gates); separation results use query complexity, while gate counts (Sec.~\ref{sec:gatecounts}) differ by $O(\log^2 C)$ per query.

\subsection{Algorithm Description}

\textbf{Algorithm 1: Fourier-NC Solver}

\textbf{Input:} Graph $G=(V,\mathcal{A})$, cycle length $C$, edge-level oracle access to cost functions $f_{ij}\colon\ZC\to\mathbb{R}$ (or equivalently, their Fourier coefficients $\hat f_{ij}(k)$, which can be computed classically in $O(C\log C)$ per edge).

\textbf{Output:} Optimal offset vector $\mu^*$.

\textbf{Step~1.} Allocate $n$ quantum registers, each of $\lceil\log_2 C\rceil$ qubits. Total: $n\cdot\lceil\log_2 C\rceil$ qubits.

\textbf{Step~2.} Prepare uniform superposition: $|{+}\rangle^{\otimes n}=(1/\sqrt{C^n})\sum_\mu|\mu\rangle$.

\textbf{Step~3.} Apply phase oracle: $U_H|\mu\rangle=e^{2\pi i\,H(\mu)/K}|\mu\rangle$, where $K=P\cdot m\cdot\max_{(i,j)}\|f_{ij}\|_\infty$ with $P=\mathrm{poly}(n,m,r)$ ensures all phases satisfy $H(\mu)/K\le 1/P$, placing the oracle in the \emph{linearisation regime} (see Fig.~\ref{fig:circuit}a for the high-level circuit). In the edge-level implementation (Sec.~\ref{sec:oracle}), each edge oracle $U_{ij}$ contributes a phase $e^{2\pi i\,f_{ij}(\mu_i-\mu_j)/K}$, so $U_H=\prod_{(i,j)}U_{ij}$.

\textbf{Step~4.} Apply inverse QFT over~$\ZCn$.

\textbf{Step~5.} Measure to obtain frequency vector~$\mathbf{k}^*$.

\textbf{Step~6.} Classical post-processing: solve system of linear congruences over $\ZC$ to reconstruct $\mu^*$ from Fourier modes.

\begin{figure}[!htbp]
\centering
\includegraphics[width=\columnwidth]{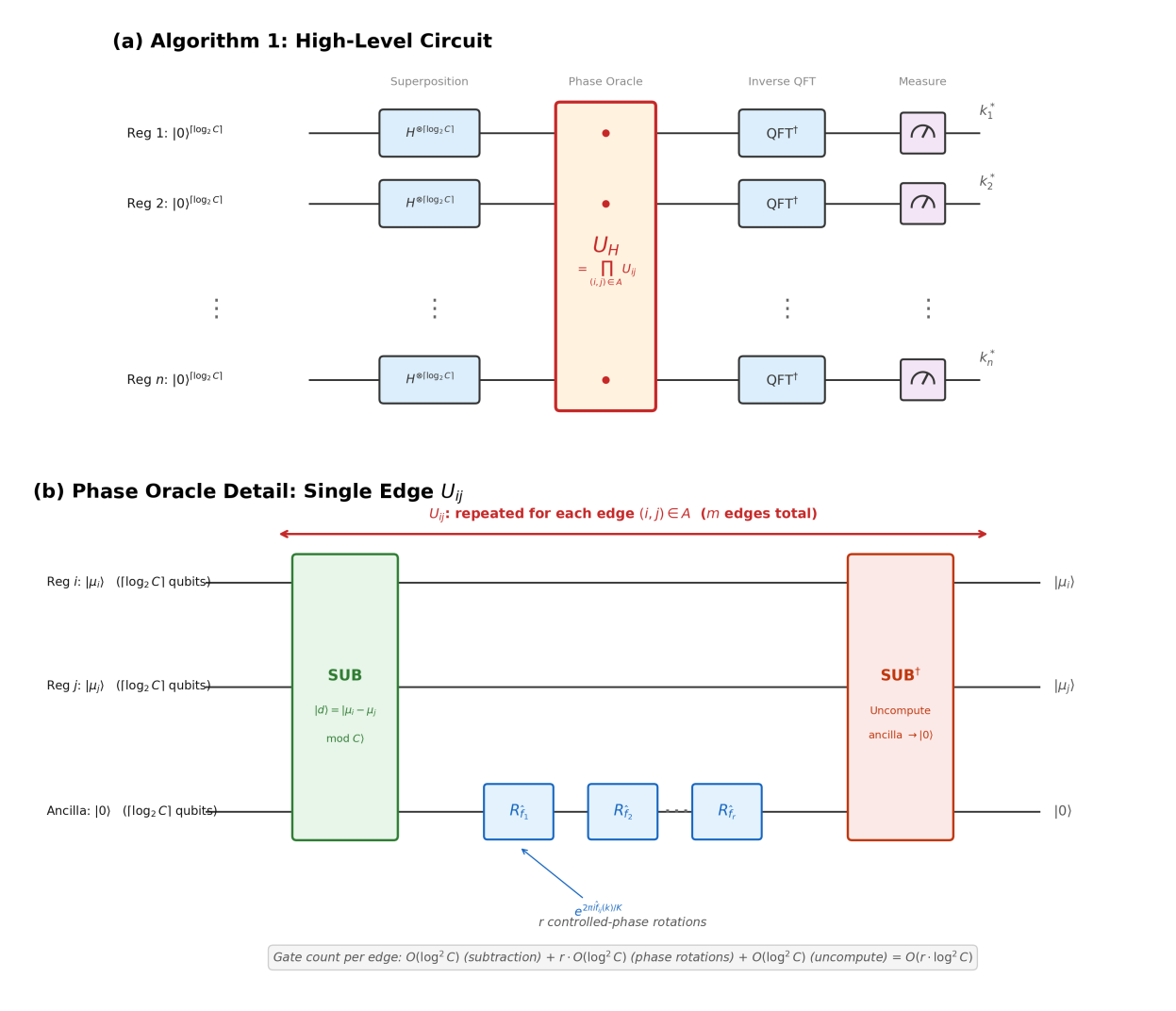}
\caption{Quantum circuit for Algorithm~1. (a)~High-level circuit: each of the $n$ vertex registers is prepared in a uniform superposition via Hadamard gates, the phase oracle $U_H$ encodes all edge constraints as a product of pairwise unitaries $U_{ij}$, and an inverse QFT extracts the Fourier-mode label $\mathbf{k}^*$. (b)~Phase oracle detail for a single edge $(i,j)$: a modular subtraction gate computes $d=\mu_i-\mu_j\bmod C$ on an ancilla register, $r$ controlled-phase rotations imprint the Fourier coefficients of the edge potential, and the subtraction is uncomputed. Gate complexity per edge is $O(r\log^2 C)$.}
\label{fig:circuit}
\end{figure}

\subsection{Phase Oracle Construction}
\label{sec:oracle}

The phase oracle $U_H=\prod_{(i,j)}U_{ij}$ factorises over edges. Each $U_{ij}$ computes $d=(\mu_i-\mu_j)\bmod C$ via a quantum subtraction circuit ($O(\log^2 C)$ gates per edge), then applies $r$ controlled-phase rotations using the Fourier synthesis $f_{ij}(d)=\sum_{k\in S_{ij}}\hat{f}_{ij}(k)\omega_C^{kd}$, with $|S_{ij}|\le r$ ($O(mr\log^2 C)$ gates total). The inverse QFT decomposes as $n$ independent QFTs over $\ZC$ ($O(n\log^2 C)$ gates). Total: $O\!\bigl((mr+n)\cdot\log^2 C\bigr)$ gates per repetition; see Fig.~\ref{fig:circuit}(b).

\subsection{Correctness}

\begin{theorem}
\label{thm:correctness}
Let $p_{\min}=\min\{|\hat H(\mathbf{k})|^2/\|\hat H\|^2:\hat H(\mathbf{k})\neq0\}$ denote the minimum non-zero Fourier weight. Algorithm~1 outputs the optimal offset with probability $\ge1-\delta$ using $T=O(P^2/p_{\min}\cdot\log(mr/\delta))$ raw measurements, where $P=\mathrm{poly}(n,m,r)$ is the linearisation parameter (Step~3). If additionally $p_{\min}\ge1/(mr\cdot n)$ (a condition that can be verified by classical preprocessing; see Remark below), the total gate complexity is $O\!\bigl(\mathrm{poly}(n,m,r,\log C,\log(1/\delta))\bigr)$.
\end{theorem}

\begin{proof}[Proof sketch (full proof in Appendix~\ref{app:correctness_proof})]
The proof proceeds in three stages. \emph{Stage~1 (Linearisation):} Setting $K=P\cdot m\cdot\max\|f_{ij}\|_\infty$ with $P\gg1$ ensures $g(\mu)=e^{2\pi iH(\mu)/K}\approx 1+(2\pi i/K)H(\mu)$, so the DFT of $g$ recovers $\hat{H}(\mathbf{k})$ up to negligible error. Conditional on measuring $\mathbf{k}\neq\mathbf{0}$, the distribution is proportional to $|\hat{H}(\mathbf{k})|^2$. \emph{Stage~2 (Linear system):} Each measured non-zero mode yields a linear congruence $k_i(\mu^*_i-\mu^*_j)\equiv\varphi_{ij}\pmod{C}$; edge-level oracle access resolves overlapping modes. Collecting congruences along a spanning tree determines all $n-1$ offset differences via the Chinese Remainder Theorem. \emph{Stage~3 (Repetition bound):} By a coupon-collector argument, $T=O(P^2/p_{\min}\cdot\log(mr/\delta))$ raw measurements suffice to collect all $s\le mr$ non-zero modes. Under $p_{\min}\ge 1/(mr\cdot n)$, this remains $O(\mathrm{poly})$. Convergence is shown in Fig.~\ref{fig:convergence}.
\end{proof}

\paragraph{Remark (on the $p_{\min}$ condition).} The condition $p_{\min}\ge1/(mr\cdot n)$ is a \emph{verifiable hypothesis}: a classical preprocessing step computes $p_{\min}$ exactly from the $O(mr)$ edge-level DFT coefficients in $O(mr\cdot C\log C)$ time. The sufficient condition is that the per-edge Fourier dynamic range $\kappa=\max_k|\hat f_{ij}(k)|/\min_{k:\hat f_{ij}(k)\neq0}|\hat f_{ij}(k)|$ be polynomially bounded; we prove $p_{\min}\ge 1/\mathrm{poly}(n,m,r)$ for cosine coupling (Proposition~\ref{prop:pmin_cosine}; covers traffic NC, OPF, Kuramoto, clock sync) and piecewise-linear costs (Proposition~\ref{prop:pmin_pwl}; covers railway PESP, TDMA/FAP, influence scheduling) in Appendix~\ref{app:pmin}. Numerical validation (Appendix~\ref{app:validation}) confirms that $p_{\min}$ exceeds the polynomial threshold by $5$--$29\times$ across all tested topologies.

\begin{figure}[!htbp]
\centering
\includegraphics[width=\columnwidth]{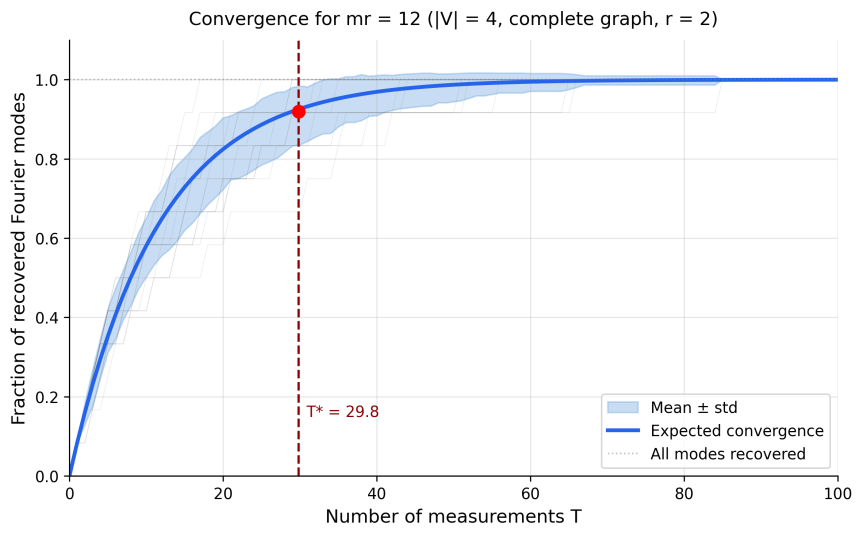}
\caption{Convergence: fraction of distinct Fourier modes recovered versus raw measurement count $T$ for the $|V|=4$ instance ($mr=12$, complete graph, $r=2$). Blue line: expected convergence; shaded band: $\pm1$ standard deviation over 50 Monte Carlo trials; grey traces: individual trials. Dashed line: theoretical threshold $T^*=mr\cdot\ln(mr)\approx29.8$.}
\label{fig:convergence}
\end{figure}

\paragraph{Oracle models.}
\label{sec:oracle_models}
An ``oracle model'' specifies what the algorithm is allowed to ask: in the \emph{global} model, it evaluates the total cost $H(\mu)$ for a complete assignment (like running a full traffic simulation); in the \emph{edge-level} model, it queries individual link costs $f_{ij}(d)$ (like measuring delay on a single road link). Whether edge-level access is practical depends on the domain: railway timetable penalties and power-flow sinusoids are given analytically, but traffic delay functions require costly per-link field experiments. For $\Sk$, even edge-level access requires evaluating functions on a domain of size $k!$, so the quantum advantage persists regardless of the oracle model.

\section{Complexity Analysis}
\label{sec:complexity}

\subsection{NP-Hardness}

\begin{theorem}
\label{thm:nphard}
The Fourier-NC optimisation problem (minimisation) is NP-hard, and the Fourier-NC Decision problem is NP-complete, even for $r=1$ and $C=2$.
\end{theorem}

\begin{proof}[Proof sketch (full proof in Appendix~\ref{app:complexity_proofs})]
Reduce from MAX-CUT: set $C=2$ and $f_{ij}(x)=\cos(\pi x)$, so $f_{ij}(0)=+1$ (same side) and $f_{ij}(1)=-1$ (different sides). Then $H(\mu)=m-2\cdot\mathrm{cut}(\mu)$, and minimising~$H$ is equivalent to MAX-CUT~\cite{Garey1976}. The cost $f_{ij}$ has $r=1$, confirming hardness at minimal Fourier sparsity. NP-completeness of the decision version follows because a candidate~$\mu$ is verifiable in $O(m)$ time. Fig.~\ref{fig:maxcut} illustrates the reduction.
\end{proof}

\begin{figure}[!htbp]
\centering
\includegraphics[width=\columnwidth]{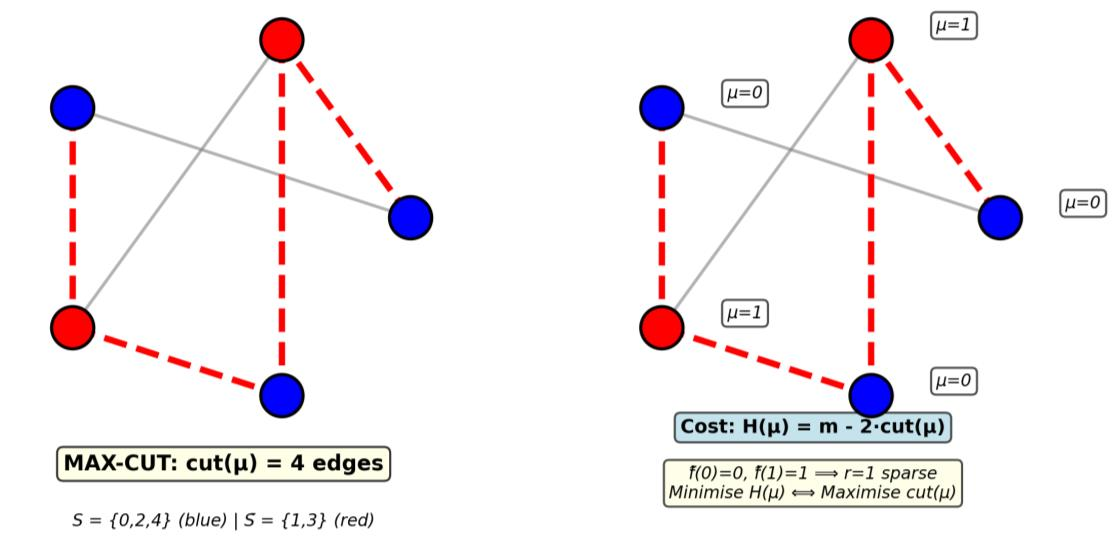}
\caption{NP-hardness reduction (Theorem~\ref{thm:nphard}). A MAX-CUT partition on a 5-node graph (left) maps directly to offset assignments $\mu\in\{0,1\}^n$ with cost $f_{ij}(\Delta\mu)=\cos(\pi\Delta\mu)$. Minimising $H(\mu)=m-2\cdot\mathrm{cut}(\mu)$ is equivalent to MAX-CUT, confirming NP-hardness at $r=1$.}
\label{fig:maxcut}
\end{figure}

Knowledge of Fourier coefficients alone does not suffice: NP-completeness holds even at $r=1$. On frustration-free graphs, the spanning-tree congruences are solvable in $O(n)$ time, so the quantum advantage lies in extracting these coefficients---an oracle-model separation.

\subsection{Classical Query Lower Bound}

\begin{theorem}[Classical query lower bound]
\label{thm:lower_bound}
Any classical algorithm in the global oracle model requires $\Omega(C^{n})$ queries to solve the \emph{general} (non-sparse) Fourier-NC problem. The adversary uses instances with $r=C$ (maximally non-sparse); for $r$-sparse instances over~$\ZC$, this lower bound does not apply (see Sec.~\ref{sec:sfft_gap}).
\end{theorem}

\begin{proof}[Proof sketch (full proof in Appendix~\ref{app:complexity_proofs})]
Construct an adversary on~$K_n$ with identical cosine costs plus a delta-function perturbation $\varepsilon\cdot\delta(\mu,\mu^\dagger)$ creating a unique minimum at random $\mu^\dagger\in\ZCn$. Any query $\mu\neq\mu^\dagger$ returns the same value regardless of~$\mu^\dagger$, so $Q\ge C^n/2$ queries are needed for success probability $>1/2$ (Fig.~\ref{fig:lower_bound}).
\end{proof}

\begin{figure}[!htbp]
\centering
\includegraphics[width=\columnwidth]{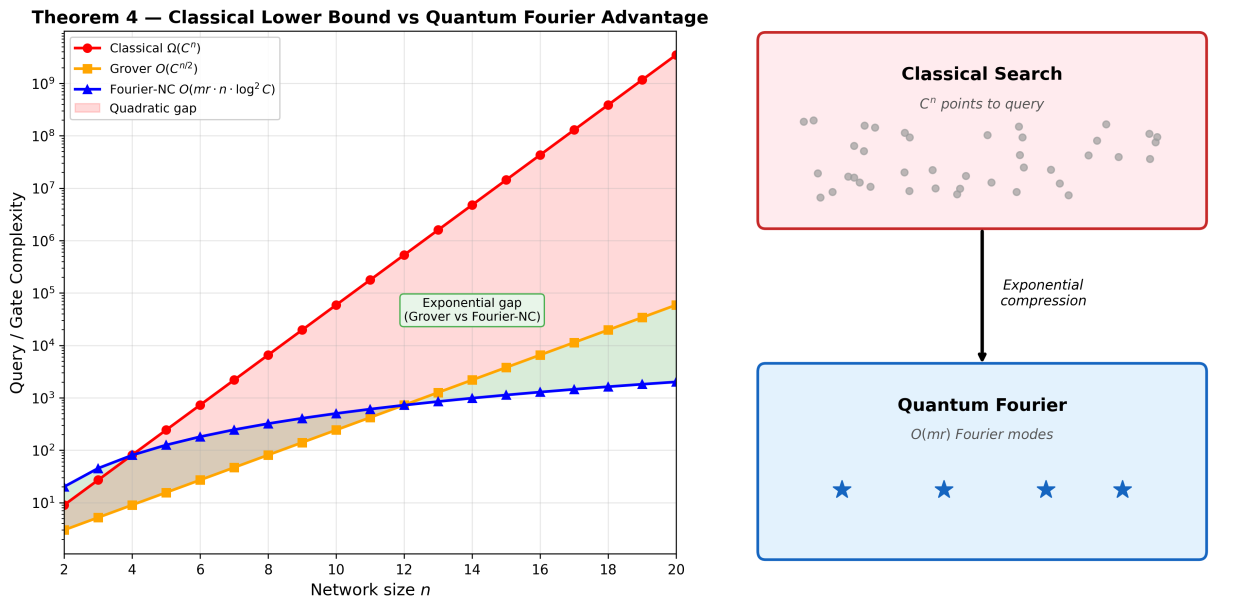}
\caption{Query complexity comparison for $C=4$. Classical algorithms require $\Omega(C^n)$ queries (red, Theorem~\ref{thm:lower_bound}). Grover achieves $O(C^{n/2})$ (orange)---still exponential. Fourier-NC achieves $O(\mathrm{poly}(n,m,r,\log C))$ (blue) by reading $O(mr)$ Fourier modes via QFT.}
\label{fig:lower_bound}
\end{figure}

\subsection{Gate Count Analysis}
\label{sec:gatecounts}

We compare gate counts for Grover and Fourier-NC using analytical formulas. Let $q=\lceil\log_2 C\rceil$. Each Fourier-NC repetition costs $G_{\mathrm{QFT}}=(mr+n)\cdot q^2$ gates, with $T=mr\cdot n\cdot\lceil\ln(mr)\rceil$ repetitions. Each Grover iteration costs $G_{\mathrm{Grover}}=m\cdot q^2$ gates, with $\lceil\pi\sqrt{C^n}/4\rceil$ iterations. Small-instance verification ($|V|\in\{3,\ldots,6\}$, $C\in\{4,8,16\}$) confirms speedups of $37$--$554\times$ per repetition. Table~\ref{tab:gates_projected} projects the scaling.

\begin{table}[!htbp]
\centering
\caption{Analytical gate count projections ($r=2$, complete graphs $m=n(n{-}1)/2$, $C=64$). Fourier-NC uses edge-level oracle; Grover uses global oracle.}
\label{tab:gates_projected}
\resizebox{\columnwidth}{!}{%
\begin{tabular}{@{}rrrrrr@{}}
\toprule
$n$ & $m$ & $C^n$ & Fourier-NC & Grover & Speedup \\
\midrule
10 & 45 & $1.2\times10^{18}$ & $1.6\times10^{7}$ & $1.4\times10^{12}$ & $8.4\times10^{4}$ \\
20 & 190 & $1.3\times10^{36}$ & $6.6\times10^{8}$ & $6.2\times10^{21}$ & $9.4\times10^{12}$ \\
50 & 1225 & $2.0\times10^{90}$ & $8.8\times10^{10}$ & $4.9\times10^{49}$ & $5.6\times10^{38}$ \\
100 & 4950 & $4.1\times10^{180}$ & $3.6\times10^{12}$ & $2.9\times10^{95}$ & $8.0\times10^{82}$ \\
\bottomrule
\end{tabular}}
\end{table}

\begin{table}[!htbp]
\centering
\caption{Complexity comparison. Fourier-NC uses edge-level oracle access; all others use the global oracle. For $r$-sparse abelian instances, classical sFFT with edge-level access matches Fourier-NC.}
\label{tab:complexity}
\resizebox{\columnwidth}{!}{%
\begin{tabular}{@{}llll@{}}
\toprule
Approach & Oracle & Complexity & Scaling \\
\midrule
Brute force & Global & $O(C^n\cdot m)$ & Exp. \\
Grover~\cite{Grover1996,DixitPech2026} & Global & $O(\sqrt{C^n}\cdot m\cdot\log^2 m)$ & Exp. \\
Classical sFFT~\cite{Hassanieh2012} & Edge & $O(m\cdot r\cdot\mathrm{polylog}\,C)$ & Poly. \\
\textbf{Fourier-NC} & \textbf{Edge} & $\mathbf{O(\mathrm{poly}(n,m,r,\log C))}$ & \textbf{Poly.} \\
\bottomrule
\end{tabular}}
\end{table}

\section{Permutation Coordination over \texorpdfstring{$\Sk$}{S\_k}}
\label{sec:symmetric}

The symmetric group $\Sk$---permutations of $k$~items---is where quantum \emph{genuinely} wins. Problems over $\Sk$ arise when agents must agree on an \emph{ordering} (ranking, scheduling, routing). Since $|\Sk|=k!$ but the largest abelian subgroup has only ${\sim}3^{k/3}$ elements, no classical abelian decomposition can work, and the QFT provides a super-exponential speedup.

\paragraph{Intuitive example: fleet route assignment.}
Consider $k$~trucks and $k$~delivery routes coordinated across $n$~depots. Each depot assigns trucks to routes---a permutation $\sigma_i\in\Sk$. Assuming a homogeneous fleet (identical capacity and type), trucks are interchangeable, so congestion between neighbouring depots depends only on the \emph{cycle type} of $\sigma_i^{-1}\sigma_j$ (how many routes stay, swap in pairs, rotate in triples), making the cost a class function. A composite cost penalising both priority violations (Kendall tau, $r=2$) and disruption (Hamming distance, $r=2$) has $r\le 4$; with unequal weights the optimal cycle type is non-trivial. For $k=15$, the $15!\approx 1.3\times10^{12}$ orderings per depot are beyond brute force, but the quantum algorithm recovers active modes in ${\sim}2.7\times10^4$ queries (Table~\ref{tab:sk_scaling}).

\subsection{The symmetric group $\Sk$}

\begin{definition}[Symmetric group]
\label{def:symmetric_group}
Every permutation in $\Sk$ ($|\Sk|=k!$) has a unique decomposition into disjoint cycles; the multiset of cycle lengths is the \emph{cycle type}, a partition of~$k$. Conjugacy classes of $\Sk$ are indexed by partitions of~$k$, with $p(k)\sim e^{\pi\sqrt{2k/3}}$ classes. The irreps of $\Sk$ are also indexed by partitions $\lambda\vdash k$; the maximum irrep dimension is $\dmax(\Sk)=\Theta(\sqrt{k!})$, super-exponential in~$k$ (see Appendix~\ref{app:irreps} for details).
\end{definition}

\subsection{Permutation Coordination problem}

\begin{definition}[Permutation Coordination]
\label{def:perm_coord}
Given a graph $G=(V,\mathcal{A})$ with $|V|=n$ agents and $|\mathcal{A}|=m$ edges, each agent~$i$ is assigned a permutation $\sigma_i\in\Sk$. Each edge $(i,j)$ has a cost function $f_{ij}\colon\Sk\to\mathbb{R}$ that depends on the relative permutation $\sigma_i^{-1}\cdot\sigma_j$. The \emph{Permutation Coordination} problem is:
\[
\minimize_{\sigma_1,\ldots,\sigma_n\in\Sk}\; H(\sigma) \;=\; \sum_{(i,j)\in\mathcal{A}} f_{ij}\!\bigl(\sigma_i^{-1}\cdot\sigma_j\bigr).
\]
A cost function $f\colon\Sk\to\mathbb{R}$ is \emph{$r$-Fourier-sparse} if at most $r$ of the $p(k)$ irreducible Fourier coefficients $\hat{f}(\lambda)$ are non-zero.
\end{definition}

This generalises Fourier-NC from cyclic ($\ZC$) and dihedral ($\DC$) groups to the symmetric group~$\Sk$. The key structural difference is that Fourier coefficients are now \emph{matrices}, not scalars.

\subsection{Cost functions}

Theorem~\ref{thm:sk_quantum} requires that edge cost functions be \emph{class functions}---constant on conjugacy classes. The two principal building blocks are: \emph{Kendall tau distance} $d_{\mathrm{KT}}$, which counts pairwise inversions and has $r=2$ (supported on trivial and standard irreps)~\cite{Kendall1938}; and \emph{Hamming distance} $d_H(\sigma)=k-\mathrm{fix}(\sigma)$, also with $r=2$. A composite cost $f(\sigma)=\sum_i w_i g_i(\sigma)$ of $\ell$ such metrics has $r\le 2\ell$. With $\ell\ge 2$ and non-uniform weights, the minimising conjugacy class depends on the weight configuration and is generally not known a priori---placing composite-cost problems in the regime where quantum advantage is genuine.\footnote{For pure Kendall tau cost ($r=2$), the identity minimiser yields $H=0$ in $O(1)$ time since $d_{\mathrm{KT}}(e)=0$ and setting $\sigma_i=\sigma$ for all~$i$ is trivially optimal; genuine advantage requires composite costs ($r\ge 4$).}

\subsection{Quantum algorithm}

The quantum algorithm follows the same Fourier-recovery structure as the abelian case (Sec.~\ref{sec:algorithm}), with the QFT over $\ZC$ replaced by the QFT over~$\Sk$.

\begin{theorem}[Quantum Permutation Coordination]
\label{thm:sk_quantum}
For the Permutation Coordination problem on $\Sk$ with $n$ agents, $m$ edges, and Fourier sparsity~$r$, where (i)~edge cost functions are class functions on~$\Sk$ and (ii)~the network graph is frustration-free, the quantum algorithm recovers all active Fourier modes and finds the optimal assignment using $O(m\cdot r\cdot\mathrm{poly}(k))$ queries to edge-level oracles.  (For non-abelian groups, frustration-freeness means that the product of edge-optimal relative permutations $d^*_{ij}$ around each cycle~$\gamma$ equals the identity in~$\Sk$, generalising the holonomy condition of Definition~\ref{def:frustration}.)
\end{theorem}

\begin{proof}[Proof sketch (full proof in Appendix~\ref{app:sk_proof})]
The proof parallels the abelian case (Theorem~\ref{thm:correctness}) with matrix-valued Fourier coefficients. By the Peter--Weyl theorem, the QFT over~$\Sk$ (implementable in $O(k^2\log k)$ gates~\cite{Beals1997}) maps each edge oracle query to a measurement over irrep labels~$\lambda\vdash k$, with probability proportional to $d_\lambda\|\hat{f}_{ij}(\lambda)\|_F^2$. For class functions, $\hat{f}(\lambda)=c_\lambda\cdot I_{d_\lambda}$, so one diagonal entry suffices per irrep; the minimisation over $p(k)$ conjugacy-class representatives (sub-exponential in~$k$) is efficient. Total query complexity: $O(m\cdot r\cdot k^2\log k)$, polynomial in~$k$ despite $|\Sk|=k!$.
\end{proof}

\paragraph{Remark (class-function restriction).} For general $r$-sparse functions on~$\Sk$, the Fourier coefficients are $d_\lambda\times d_\lambda$ matrices with $d_\lambda$ up to $\Theta(\sqrt{k!})$, making the minimisation step exponential; whether sparsity implies sufficient concentration for efficient approximate minimisation remains open.

\paragraph{Remark (query vs.\ total complexity).}
Theorem~\ref{thm:sk_quantum} counts \emph{oracle queries}; the classical minimisation over $p(k)\sim e^{\pi\sqrt{2k/3}}$ conjugacy classes adds sub-exponential post-processing ($e^{O(\sqrt{k})}$ total time). This is far below $O(k!)$ but not polynomial, so the full algorithm is not in BQP in general. Corollary~\ref{cor:bqp} identifies the subclass for which it is.\footnote{The QFT over~$\Sk$~\cite{Beals1997} uses $O(k^2\log k)$ gates; compiling each to precision $\epsilon_g=O(1/(k^2\log k))$ via Solovay--Kitaev~\cite{NielsenChuang2000} adds $O(\log(k/\epsilon_g))$ T-gates per rotation, without changing asymptotic query complexity.}

\subsection{Determined-minimiser instances and complexity classification}

\begin{definition}[Determined-Minimiser Permutation Coordination (DMPC)]
\label{def:dmpc}
A Permutation Coordination instance (Definition~\ref{def:perm_coord}) is a \emph{determined-minimiser} instance if, for each edge $(i,j)$, the minimising conjugacy class of $f_{ij}$ can be identified in $\mathrm{poly}(k)$ time---for example, when the cost is a composite class-function metric whose weight configuration determines the optimal cycle type. The DMPC \emph{search problem} asks for the optimal assignment $\sigma^*$; the DMPC \emph{decision problem} asks whether $H(\sigma)\le\tau$ for a given threshold.
\end{definition}

\begin{corollary}[BQP membership for DMPC]
\label{cor:bqp}
Under the conditions of Theorem~\ref{thm:sk_quantum}, for a DMPC instance (Definition~\ref{def:dmpc}), the full quantum algorithm (Fourier recovery $+$ minimisation $+$ spanning-tree propagation) runs in $O(m\cdot r\cdot\mathrm{poly}(k))$ total time and the problem is in~\textnormal{BQP}. Under Assumption~\ref{ass:no_sfft}, it is not in~\textnormal{P}; the quantum speedup $k!\to\mathrm{poly}(k)$ is genuine for the search problem.
\end{corollary}

\begin{proof}
Fourier recovery: $O(m\cdot r\cdot k^2\log k)$ queries, each specifiable in $O(k\log k)$ bits and executable with $O(k^2\log k)$ gates, so total \emph{time} (not merely query count) is $O(m\cdot r\cdot\mathrm{poly}(k))$. When the minimising conjugacy class is known \emph{a priori}, the $p(k)$-search bottleneck vanishes; spanning-tree propagation adds $O(n\cdot k)$. Total: $O(m\cdot r\cdot\mathrm{poly}(k))$, polynomial in all parameters, placing the search problem in BQP. For the \emph{decision} version (``does there exist $\sigma$ with $H(\sigma)\le\tau$?''), a candidate $\sigma$ is verifiable in $O(mk\log k)$ time, so DMPC Decision is in NP. Assumption~\ref{ass:no_sfft} implies no classical algorithm can match this with polynomially many queries, hence the search problem is not in~P.
\end{proof}

\paragraph{Remark (complexity classification of DMPC).}
\label{rem:complexity_class}
DMPC lies in NP$\,\cap\,$BQP and is conditionally outside~P (under Assumption~\ref{ass:no_sfft}). It is almost certainly not NP-hard, since NP-hard$\,\cap\,$BQP would imply the widely-disbelieved NP$\,\subseteq\,$BQP. This places DMPC in the intermediate complexity regime alongside integer factorisation and graph isomorphism. Whether DMPC$\,\in\,$co-NP remains open. A natural candidate for a co-NP witness---membership of the minimiser in a $\mathrm{poly}(r,k)$-sized ``extremal'' set of conjugacy classes---fails computationally: the minimiser of a random $r$-sparse class function can require up to $\Theta(\sqrt{k})$ distinct cycle-type parts regardless of~$r$ (Appendix~\ref{app:ecc_verification}).

\paragraph{Remark (applicability).}
Two situations qualify for DMPC: (i)~the minimiser is uniquely determined by symmetry (e.g.\ the identity class for pure Kendall tau, though this case is classically trivial); (ii)~domain knowledge constrains the search to $O(\mathrm{poly}(k))$ candidate cycle types. Establishing that specific composite metrics admit polynomial-time minimiser identification is an open problem.\footnote{Kemeny ranking aggregation (NP-hard even for four inputs~\cite{Bartholdi1989}) optimises one consensus ranking among fixed inputs. Permutation Coordination instead assigns variable permutations on a graph; Corollary~\ref{cor:bqp} identifies a tractable subclass, not a general Kemeny solver.}

\subsection{Classical hardness and exponential separation}

\begin{assumption}[No efficient classical sFFT for $\Sk$]
\label{ass:no_sfft}
There is no classical algorithm that recovers the $r$ active irrep labels and Fourier coefficients of an $r$-sparse class function $f\colon\Sk\to\mathbb{R}$ using $\mathrm{poly}(k,r)$ oracle evaluations.
\end{assumption}

\paragraph{Evidence for Assumption~\ref{ass:no_sfft}.}
Two lines of evidence support this assumption. \emph{(i)~Connection to the Hidden Subgroup Problem.} The $\Sk$-HSP is at least as hard as graph isomorphism~\cite{MooreRussellSchulman2005}. An efficient classical sFFT for~$\Sk$ would solve the coset-distinguishing variant. Crucially, even \emph{quantum} strong Fourier sampling fails for the $\Sk$-HSP~\cite{MooreRussellSchulman2005}. \emph{(ii)~Structural barriers.} Classical sFFT for abelian groups~\cite{Hassanieh2012,KushilevitzMansour1993} relies on subgroup subsampling and hashing; both fail for~$\Sk$ because the largest abelian subgroup has order ${\sim}3^{k/3}\ll k!$ and the non-commutativity prevents any known hashing analogue.

\begin{proposition}[Conditional super-exponential separation for $\Sk$]
\label{thm:sk_separation}
Combining Theorem~\ref{thm:sk_quantum} with Assumption~\ref{ass:no_sfft}: the quantum algorithm recovers $r$-sparse class-function Fourier coefficients over~$\Sk$ using $O(m\cdot r\cdot k^2\log k)$ oracle evaluations (polynomial in~$k$), while any classical algorithm requires $\omega(\mathrm{poly}(k))$ evaluations under the assumption. The conditional speedup is $k!\to\mathrm{poly}(k)$---super-exponential in~$k$.
\end{proposition}

Table~\ref{tab:sk_scaling} and Fig.~\ref{fig:sk_scaling} show the query complexity comparison.

\begin{table}[!htbp]
\centering
\caption{Query complexity for Permutation Coordination over $\Sk$ ($m=10$ edges, $r=3$). Quantum: edge-level queries ($m\cdot r\cdot k^2\lceil\log_2 k\rceil$); Classical: function evaluations ($m\cdot k!$). Crossover occurs at $k\approx 6$.}
\label{tab:sk_scaling}
\resizebox{\columnwidth}{!}{%
\begin{tabular}{@{}rrrrl@{}}
\toprule
$k$ & $|S_k|$ & Q.\ queries & Cl.\ queries & Speedup \\
\midrule
3 & 6 & $5.4\times10^2$ & $6.0\times10^1$ & $0.1\times$ \\
5 & 120 & $2.3\times10^3$ & $1.2\times10^3$ & $0.5\times$ \\
7 & 5\,040 & $4.4\times10^3$ & $5.0\times10^4$ & $11\times$ \\
10 & 3\,628\,800 & $1.2\times10^4$ & $3.6\times10^7$ & $3\,024\times$ \\
15 & $1.3\times10^{12}$ & $2.7\times10^4$ & $1.3\times10^{13}$ & $4.8\times10^8$ \\
\bottomrule
\end{tabular}}
\end{table}

\begin{figure}[!htbp]
\centering
\includegraphics[width=\columnwidth]{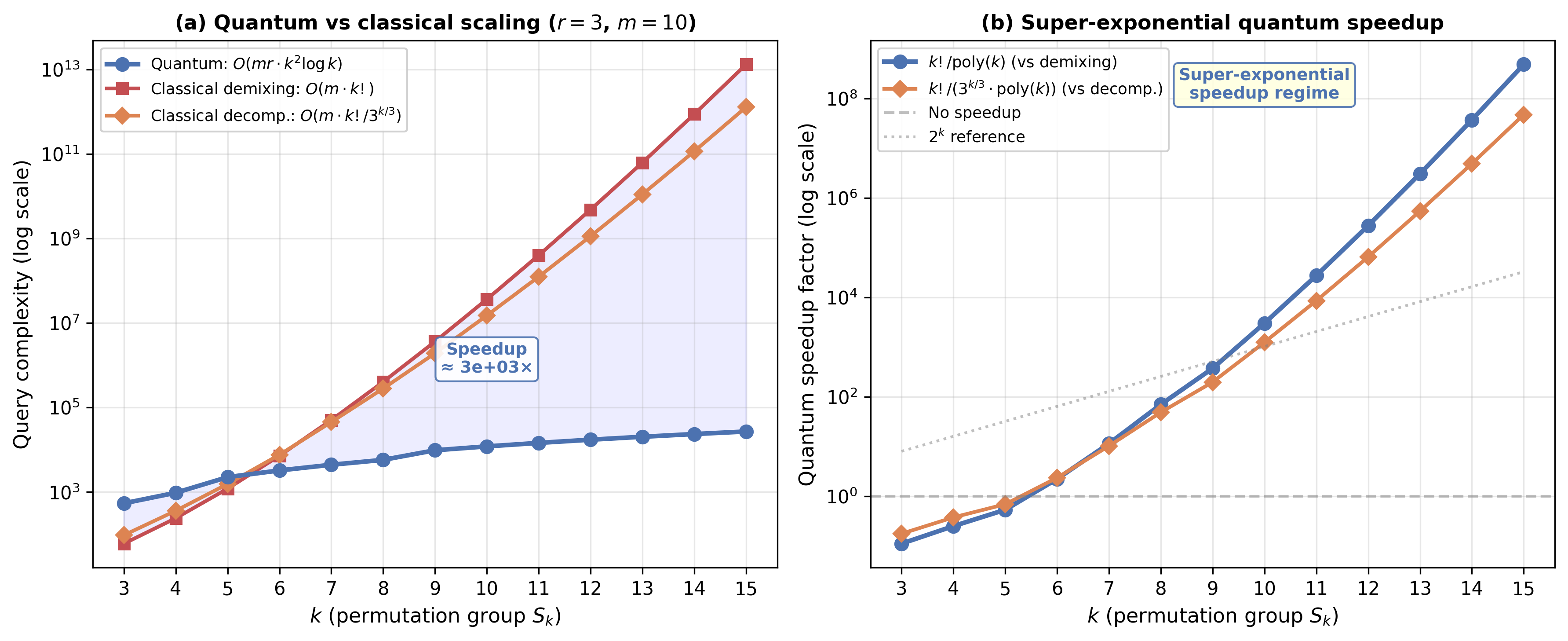}
\caption{Quantum vs.\ classical scaling for Permutation Coordination over $\Sk$ ($r=3$, $m=10$). (a)~Query complexity: quantum scales as $\mathrm{poly}(k)$ while classical scales as $k!$. (b)~The speedup ratio is super-exponential, exceeding $2^k$ (dashed) for all $k\ge7$.}
\label{fig:sk_scaling}
\end{figure}

\subsection{Structural classification and classical limits}
\label{sec:sfft_gap}

We now unify the abelian--non-abelian landscape through a single structural invariant.

\begin{definition}[Abelian index]
\label{def:abelian_index}
For a finite group~$G$, the \emph{abelian index} is $\alpha(G)=|G|/\max\{|A|:A\le G,\;A\text{ abelian}\}$.
\end{definition}

For the three groups studied: $\alpha(\ZC)=1$, $\alpha(\DC)=2$, and $\alpha(\Sk)=k!/3^{k/3}$~\cite{DixonMortimer1996}.

\begin{proposition}[Fundamental inequality]
\label{prop:fundamental_ineq}
For any finite group~$G$, $\dmax(G)\le\alpha(G)$.
\end{proposition}

\begin{proof}
Let $A\le G$ be abelian with $[G:A]=\alpha(G)$, and let $\rho\in\mathrm{Irr}(G)$ achieve $d_\rho=\dmax(G)$. Since $A$ is abelian, $\mathrm{Res}^G_A\rho$ decomposes into one-dimensional representations; by Frobenius reciprocity, $\rho$ embeds in $\mathrm{Ind}^G_A\chi$ for some $\chi\in\mathrm{Irr}(A)$, giving $\dmax(G)\le[G:A]=\alpha(G)$.
\end{proof}

\paragraph{The abelian complexity trichotomy.}
Proposition~\ref{prop:fundamental_ineq} yields three regimes. \emph{Regime~I} ($\alpha(G)=1$): $G$ is abelian; classical sFFT~\cite{Hassanieh2012} recovers all $mr$ coefficients from $O(m\cdot r\cdot\mathrm{polylog}\,C)$ edge-level queries, matching quantum. \emph{Regime~II} ($1<\alpha(G)=\dmax(G)$, bounded): classical coset decomposition gives polynomial cost; advantage is at most polynomial (includes~$\DC$, $\alpha=2$). \emph{Regime~III} ($\alpha(G)\gg\dmax(G)\gg1$): no efficient coset decomposition exists; quantum advantage is genuine. $\Sk$ sits firmly in Regime~III.

Even in the global oracle model, classical \emph{node-sweep} demixing extracts edge-level coefficients at cost $O(m\cdot C)$ for abelian groups---polynomial, though independent of~$r$. For~$\Sk$, demixing costs $O(m\cdot k!)$---exponential---while quantum runs in $O(m\cdot r\cdot\mathrm{poly}(k))$.

\begin{table*}[!htbp]
\centering
\caption{Group progression: structural parameters, complexity, and quantum speedup. Abbreviations: f.f.\ = frustration-free; cl.fn.\ = class function.}
\label{tab:progression}
\resizebox{\textwidth}{!}{%
\begin{tabular}{@{}lllll@{}}
\toprule
Group & Abel.\ idx & sFFT? & Complexity & Q.\ speedup \\
\midrule
$\ZC$ & 1 & Yes & NP-c (gen.); P$\,\cap\,$BQP ($r$-sparse, f.f.) & Poly. \\
$\DC$ & 2 & Yes & NP-c (gen.); P$\,\cap\,$BQP ($r$-sparse, f.f.) & Poly. \\
$\Sk$ DMPC & $k!/3^{k/3}$ & \textbf{No} & NP$\,\cap\,$BQP; cond.\ $\notin$~P (Cor.~\ref{cor:bqp}, Asm.~\ref{ass:no_sfft}). Not NP-hard. & $k!\to\mathrm{poly}(k)$ \\
$\Sk$ (gen.) & $k!/3^{k/3}$ & \textbf{No} & Poly.\ queries, $e^{O(\sqrt{k})}$ total; classically EXPTIME & $k!\to e^{O(\sqrt{k})}$ \\
\bottomrule
\end{tabular}}
\end{table*}

\subsection{Structural complexity boundary}
\label{sec:complexity_dichotomy}

The algorithm does not solve all Fourier-NC instances (which would imply NP$\,\subseteq\,$BQP). The boundary is determined by frustration-freeness (Definition~\ref{def:frustration}). At $C=2$, $r=1$, the algorithm succeeds precisely on bipartite graphs---where MAX-CUT is in~P. The instances where the algorithm succeeds are therefore classically tractable with edge-level access, placing abelian Fourier-NC in the HHL category~\cite{HHL2009}: the quantum speedup is real within the oracle framework but does not put NP-complete problems in BQP. For~$\Sk$ DMPC instances, the situation is different: the problem is conditionally outside~P even with edge-level access (Remark~\ref{rem:complexity_class}), placing it in the intermediate regime. A hybrid approach---quantum Fourier recovery plus classical enumeration over $C^\beta$ residual configurations---is exact when $\beta=O(\log n/\log C)$, typical of engineering networks~\cite{Barthelemy2011}. The complexity landscape (Fig.~\ref{fig:complexity_landscape}) reflects this structure.

\begin{figure}[!htbp]
\centering
\includegraphics[width=\columnwidth]{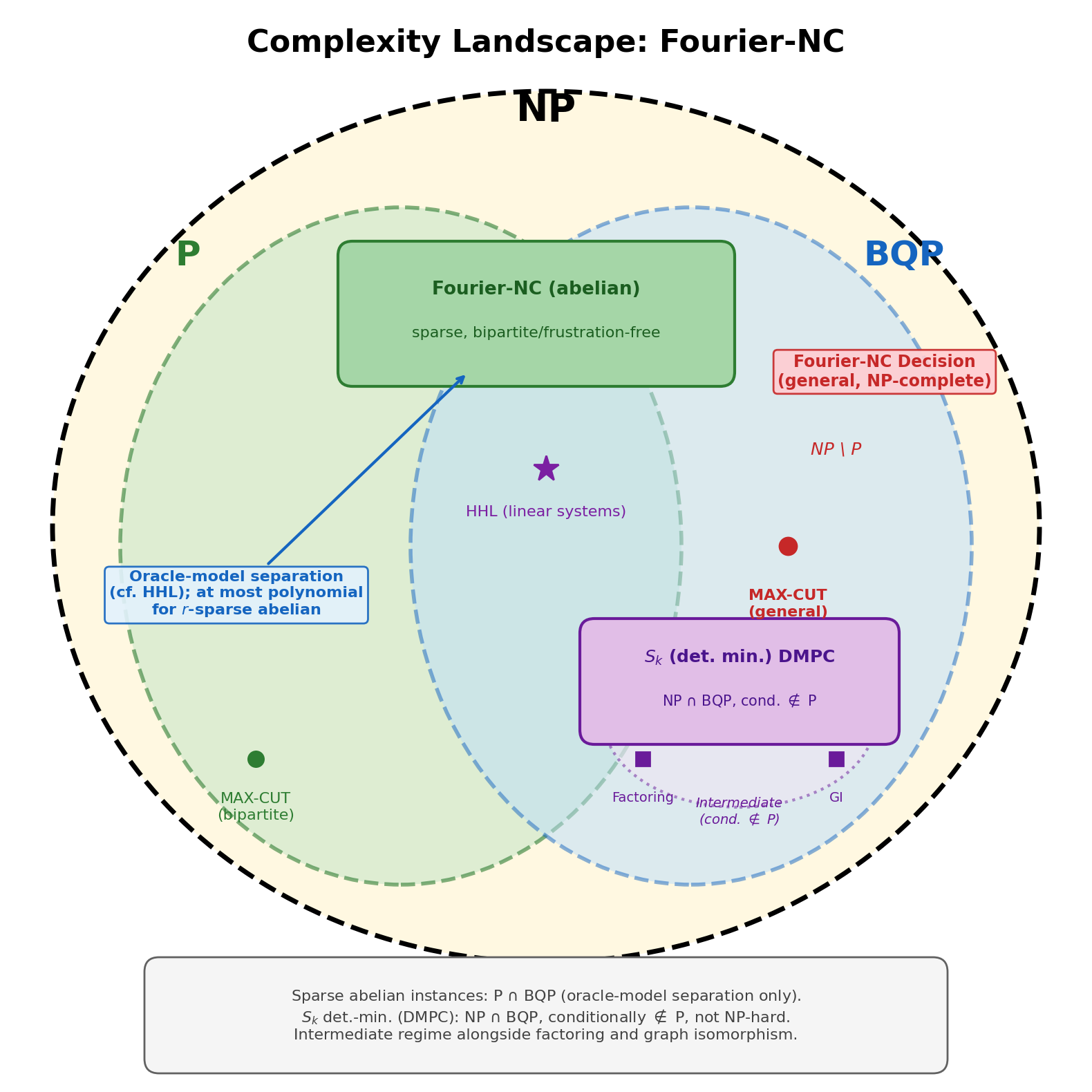}
\caption{Computational complexity landscape. $\ZC^n$/$\DC$ general instances are NP-complete; their $r$-sparse frustration-free restrictions lie in P$\,\cap\,$BQP. DMPC (Cor.~\ref{cor:bqp}) lies in NP$\,\cap\,$BQP and is conditionally outside~P (dashed boundary), alongside integer factorisation and graph isomorphism. General $r$-sparse class-function instances lie in NP but are not known to be in BQP due to $e^{O(\sqrt{k})}$ minimisation.}
\label{fig:complexity_landscape}
\end{figure}

\subsection{Limitations}
\label{sec:limitations}

The principal limitations are: (i)~for $r$-sparse abelian instances, classical sFFT matches quantum; genuine advantage requires~$\Sk$; (ii)~the $\Sk$ separation is conditional on Assumption~\ref{ass:no_sfft}; (iii)~polynomial runtime requires $p_{\min}\ge 1/\mathrm{poly}$, verified for standard cost functions but not adversarial instances; (iv)~frustrated instances yield bounded approximation (Proposition~\ref{prop:frustration_gap}); (v)~coherent QFT circuits of depth $O(k^2\log k)$ may exceed NISQ capabilities; (vi)~extending beyond class functions on~$\Sk$ remains open.

\section{Conclusion}
\label{sec:conclusion}

We introduced the Fourier Network Coordination problem (Fourier-NC), unifying network coordination across 8 application domains through a shared Fourier-sparse pairwise structure. For abelian and dihedral groups, classical sFFT matches quantum for sparse instances; the genuine separation emerges for~$\Sk$: a conditional super-exponential speedup of $k!\to\mathrm{poly}(k)$ for class-function costs with non-trivial minimisers, placing DMPC in NP$\,\cap\,$BQP and conditionally outside~P (Corollary~\ref{cor:bqp}).

The central structural insight is that the \emph{abelian index} $\alpha(G)=[G:A_{\max}]$ governs the quantum--classical gap: when $\alpha(G)$ is constant ($\ZC$, $\DC$), classical sFFT suffices; when super-exponential ($\Sk$), quantum coherence is irreplaceable. This establishes non-commutativity---not merely exponential state-space size---as the structural prerequisite for super-polynomial quantum advantage in network coordination.

Future work includes unconditional lower bounds for sparse Fourier recovery over~$\Sk$, implementation on fault-tolerant hardware, and extensions to other non-abelian groups (wreath products $\mathbb{Z}_m\wr\Sk$, general linear groups $\mathrm{GL}_n(\mathbb{F}_q)$) along the abelian index frontier.

\appendix
\numberwithin{equation}{section}
\numberwithin{table}{section}
\numberwithin{figure}{section}
\numberwithin{theorem}{section}


\section{Extension to Directed Networks}
\label{app:dihedral}

When edge costs are direction-dependent ($f_{ij}\neq f_{ji}$), the symmetry group extends from $\ZC$ to the dihedral group $\DC=\langle r,s\mid r^C=s^2=e,\;srs=r^{-1}\rangle$ of order~$2C$.

\begin{theorem}[Dihedral Complexity and Correctness]
\label{thm:dihedral_complexity}
For directed NC on $\DC$ with $r$-sparse costs, the quantum algorithm runs in $O\!\bigl((4mr+n)\cdot\log^2 C\bigr)$ gates per repetition and recovers the optimal assignment with probability $\ge1-\delta$ using $T=O(mr\cdot n\cdot\log(mr/\delta))$ repetitions.
\end{theorem}

\begin{proof}[Proof sketch]
The QFT over $\DC$ decomposes into a QFT over $\ZC$ plus a Hadamard on a reflection qubit. All irreps of $\DC$ have dimension $\le2$, so each edge requires $4r$ controlled-rotation gates. The $4\times$ gate overhead does not affect the polynomial scaling.
\end{proof}

\section{Full Proof of Theorem~\ref{thm:correctness} (Abelian Correctness)}
\label{app:correctness_proof}

\begin{proof}
The proof proceeds in three stages: coefficient recovery, linear system solution, and error analysis.

\emph{Stage~1 (Coefficient Recovery via Linearisation).} After applying $H^{\otimes n\cdot\log_2 C}$ and the phase oracle $U_H$, the quantum state is $|\psi\rangle=(1/\sqrt{C^n})\sum_\mu g(\mu)|\mu\rangle$ where $g(\mu)=e^{2\pi i\,H(\mu)/K}$. Because the phase oracle encodes $g=e^{2\pi iH/K}$ rather than $H$ directly, the DFT of $g$ is not the DFT of $H$---the exponential introduces all harmonics via the Jacobi--Anger expansion. To recover the Fourier coefficients of~$H$, we operate in the \emph{linearisation regime}: setting $K=P\cdot m\cdot\max\|f_{ij}\|_\infty$ with $P=\mathrm{poly}(n,m,r)$ ensures $|2\pi H(\mu)/K|\le 2\pi/P\ll 1$ for all~$\mu$. Taylor expanding: $g(\mu)=1+\tfrac{2\pi i}{K}H(\mu)+O(1/P^2)$. Taking the DFT: for $\mathbf{k}\neq\mathbf{0}$, $\hat{g}(\mathbf{k})=\tfrac{2\pi i}{K}\hat{H}(\mathbf{k})+O(1/P^2)$; the zero mode is $\hat{g}(\mathbf{0})=1+O(1/P)$. Most measurement probability ($1-O(1/P^2)$) concentrates on the zero mode~$\mathbf{k}=\mathbf{0}$, but \emph{conditional on measuring $\mathbf{k}\neq\mathbf{0}$}, the distribution is proportional to $|\hat{H}(\mathbf{k})|^2+O(1/P^2)$. Since the total error across all $s\le mr$ non-zero modes is $O(mr/P^2)$, correct sampling requires $P\gg\sqrt{mr/p_{\min}}$; under $p_{\min}\ge 1/(mr\cdot n)$, we set $P=\Theta((mr)^2\cdot n)$, which remains polynomial.

\emph{Stage~2 (Linear System).} The phase oracle factorises as $U_H=\prod_{(i,j)}U_{ij}$ (Sec.~\ref{sec:oracle}). By Theorem~\ref{thm:fourier_fact}, the only non-zero Fourier modes of $H$ lie on the anti-diagonal $k_j=-k_i\bmod C$ for each edge $(i,j)$. When measurement yields frequency vector $\mathbf{k}$ with exactly one non-zero pair $(k_i,k_j=-k_i)$, this identifies a specific edge $(i,j)$ and its Fourier mode. Edge-level oracle access allows independent recovery of each $\hat{f}_{ij}(k_i)$, so the per-edge phase $\varphi_{ij}=\arg\hat f_{ij}(k_i)$ is determined without ambiguity. Each measured $\mathbf{k}$ reveals the linear congruence $k_i\cdot(\mu^*_i-\mu^*_j)\equiv\varphi_{ij}\pmod{C}$. When $\gcd(k_i,C)>1$, multiple congruences with coprime frequencies resolve ambiguity via the Chinese Remainder Theorem. Collecting congruences along a spanning tree determines all $n-1$ offset differences.

\emph{Stage~3 (Repetition Bound).} Due to linearisation, each raw measurement returns the zero mode with probability $1-O(1/P^2)$; conditional on non-zero outcome, $\Pr[\mathbf{k}]\propto|\hat H(\mathbf{k})|^2/\|\hat H\|^2$. By a coupon-collector argument, $T_{\mathrm{raw}}=O(P^2/p_{\min}\cdot\log(mr/\delta))$ measurements suffice. Under $p_{\min}\ge 1/(mr\cdot n)$, the total is $O(\mathrm{poly})$. Phase estimation is exact because phases take discrete values in $\{0,\ldots,C{-}1\}$.
\end{proof}

\section{Proof of the Fourier Factorization Theorem}
\label{app:fourier_fact_proof}

\begin{proof}
By linearity, $\hat H(\mathbf{k}) = \sum_{(i,j)\in\mathcal{A}} \hat H_{ij}(\mathbf{k})$, where $\hat H_{ij}(\mathbf{k}) = C^{-n}\sum_{\boldsymbol{\mu}} f_{ij}(\mu_i{-}\mu_j)\,\omega_C^{-\mathbf{k}\cdot\boldsymbol{\mu}}$. For $\ell\neq i,j$, the sum over $\mu_\ell$ gives $\sum_{\mu_\ell}\omega_C^{-k_\ell\mu_\ell}=C\cdot\delta_{k_\ell,0}$ (character orthogonality), forcing $k_\ell=0$. Substituting $d=\mu_i{-}\mu_j$, $t=\mu_j$:
\begin{multline}
\hat H_{ij}(\mathbf{k}) = \frac{1}{C^2}\sum_{d} f_{ij}(d)\,\omega_C^{-k_i d}\sum_{t}\omega_C^{-(k_i+k_j)t} \\
= \delta_{k_j,-k_i}\cdot\hat f_{ij}(k_i).
\label{eq:dft_factored}
\end{multline}
Each edge contributes at most $r$ non-zero frequency vectors. Summing: at most $mr$ non-zero modes total.
\end{proof}

\section{$p_{\min}$ Bounds for Engineering Cost Functions}
\label{app:pmin}

\begin{proposition}[$p_{\min}$ bound for cosine coupling]
\label{prop:pmin_cosine}
Let each edge cost function have the form $f_{ij}(x)=\sum_{\ell=1}^{r'}a_{ij,\ell}\cos(2\pi\ell\, x/C)$ with $a_{ij,\ell}>0$, and let $\Delta=\max_i\deg(i)$. Then $p_{\min} \ge 1/(\Delta\cdot m\cdot r'\cdot\kappa^2)$, where $\kappa = \max a/\min a$ is the dynamic range. For sparse networks ($\Delta=O(1)$) with $\kappa=O(\mathrm{poly}(n))$, this gives $p_{\min}\ge 1/\mathrm{poly}(n,m,r)$.
\end{proposition}

\begin{proof}
The DFT of $a\cos(2\pi\ell x/C)$ gives $|\hat{f}(\ell)|=a/2$, so $\min|\hat{f}|^2\ge(\min a)^2/4$. By Theorem~\ref{thm:fourier_fact}, at most $\deg(i)$ edges share vertex~$i$, so $\|\hat{H}\|^2\le\Delta\cdot m\cdot\sum_\ell(\max a_\ell)^2/4$. Dividing yields the bound.
\end{proof}

\begin{proposition}[$p_{\min}$ bound for piecewise-linear costs]
\label{prop:pmin_pwl}
For piecewise-linear periodic costs with total variation $V_f$, Fourier coefficients satisfy $|\hat{f}(k)|=O(V_f/k^2)$. Truncating to $r'$ harmonics gives dynamic range $\kappa\le(r')^2$, yielding $p_{\min}\ge 1/\mathrm{poly}(n,m,r)$.
\end{proposition}

\section{Proofs of NP-Hardness and Classical Query Lower Bound}
\label{app:complexity_proofs}

\subsection*{Proof of Theorem~\ref{thm:nphard} (NP-Hardness)}

\begin{proof}
We reduce from MAX-CUT. Given a graph $G=(V,\mathcal{A})$, construct a Fourier-NC instance with $C=2$ and $f_{ij}(x)=\cos(\pi x)$. Then $f_{ij}(0)=+1$ and $f_{ij}(1)=-1$, so $H(\mu)=m-2\cdot\mathrm{cut}(\mu)$. Minimising $H$ is equivalent to MAX-CUT~\cite{Garey1976}. NP-completeness of the decision version follows because a candidate $\mu$ is verifiable in $O(m)$ time.
\end{proof}

\subsection*{Proof of Theorem~\ref{thm:lower_bound} (Classical Query Lower Bound)}

\begin{proof}
Take $K_n$ with identical cosine costs plus $\varepsilon\cdot\delta(\mu,\mu^\dagger)$. For $\mu\neq\mu^\dagger$, the oracle returns the same value regardless of~$\mu^\dagger$, so $Q\ge C^n/2$ queries are needed. Grover achieves $O(C^{n/2})$---still exponential.
\end{proof}

\section{Numerical Validation}
\label{app:validation}

We validated $p_{\min}$ bounds and mode recovery on five topologies ($C=32$, 100 Monte Carlo trials). Table~\ref{tab:validation} reports results.

\begin{table}[!htbp]
\centering
\caption{Numerical validation ($C=32$, 100 trials). Bound: $1/(n\cdot m\cdot r)$.}
\label{tab:validation}
\resizebox{\columnwidth}{!}{%
\begin{tabular}{@{}llrrrrr@{}}
\toprule
Topology & Cost & $n$ & $m$ & $p_{\min}$ & Bound & Avg.\ meas.\ \\
\midrule
$4{\times}4$ grid & cos, $r{=}2$ & 16 & 24 & 0.038 & $1.3{\times}10^{-3}$ & 37 \\
8-ring & cos, $r{=}2$ & 8 & 8 & 0.045 & $7.8{\times}10^{-3}$ & 36 \\
$K_8$ & cos, $r{=}2$ & 8 & 28 & 0.035 & $2.2{\times}10^{-3}$ & 41 \\
Barbell $2{\times}K_5$ & cos, $r{=}2$ & 10 & 21 & 0.033 & $2.4{\times}10^{-3}$ & 44 \\
8-ring (PESP) & PWL, $r{=}4$ & 8 & 8 & 0.029 & $3.9{\times}10^{-3}$ & 50 \\
\bottomrule
\end{tabular}}
\end{table}

Across all topologies, $p_{\min}$ exceeds the polynomial threshold by $5$--$29\times$ and mode recovery converges within the coupon-collector bound.

\section{Full Proof of Theorem~\ref{thm:sk_quantum}}
\label{app:sk_proof}

\begin{proof}
The proof proceeds in three stages, paralleling the abelian proof but with matrix-valued Fourier coefficients.

\emph{Stage~1 (Quantum: Fourier coefficient recovery via linearisation).}
By the Peter--Weyl theorem, $f(\sigma)=\sum_{\lambda\vdash k}d_\lambda\,\mathrm{tr}(\hat{f}(\lambda)\,\rho^\lambda(\sigma))$. Group elements are encoded in $O(k\log k)$ qubits via the Lehmer code. The QFT over~$\Sk$~\cite{Beals1997} maps $|\sigma\rangle\mapsto\sum_\lambda\sqrt{d_\lambda/k!}\sum_{p,q}\rho^\lambda_{pq}(\sigma)|\lambda,p,q\rangle$ in $O(k^2\log k)$ gates. Linearisation ($K_{ij}=P\cdot\|f_{ij}\|_\infty$) gives $\hat{g}(\lambda)\approx\delta_{\lambda,\mathrm{triv}}+(2\pi i/K_{ij})\hat{f}_{ij}(\lambda)+O(1/P^2)$. For class functions, $\hat{f}(\lambda)=c_\lambda I_{d_\lambda}$, so the linearised coefficients are scalars. This stage is inherently quantum (Assumption~\ref{ass:no_sfft}).

\emph{Stage~2 (Classical: optimisation and propagation).}
Given the $r$ active irrep labels and coefficients $c_\lambda$, evaluate $f_{ij}(\sigma)=\sum_\lambda c_\lambda\chi^\lambda(\sigma)$ on $p(k)$ conjugacy-class representatives ($e^{O(\sqrt{k})}$ evaluations). Propagate the global optimum along a spanning tree.

\emph{Stage~3 (Complexity accounting).}
Total quantum query complexity: $O(m\cdot r\cdot k^2\log k)$. Classical minimisation adds $O(m\cdot r\cdot p(k))$ operations; since $p(k)\sim e^{\pi\sqrt{2k/3}}$, total time is $e^{O(\sqrt{k})}$---far below $O(k!)$.
\end{proof}

\section{Irreducible Representations of $\Sk$}
\label{app:irreps}

The irreps of $\Sk$ are indexed by partitions $\lambda\vdash k$ (Young diagrams). The dimension of irrep $\rho^\lambda$ is given by the \emph{hook length formula}~\cite{FrameRobinsonThrall1954}:
\begin{equation}
\label{eq:hook_length}
\dim(\rho^\lambda) \;=\; \frac{k!}{\prod_{(i,j)\in\lambda} h(i,j)},
\end{equation}
where $h(i,j)=(\lambda_i - j)+(\lambda'_j - i)+1$ is the hook length at box $(i,j)$. For example, $S_4$ has five irreps with dimensions $1,3,2,3,1$, satisfying $\sum d_\lambda^2=4!$. The maximum dimension grows as $\Theta(\sqrt{k!})$, super-exponentially; combined with $p(k)\sim e^{\pi\sqrt{2k/3}}$ conjugacy classes, the representation-theoretic complexity of~$\Sk$ vastly exceeds that of any abelian group. The largest abelian subgroup has order ${\sim}3^{k/3}$\footnote{Generated by disjoint cycles of length~3; see Dixon and Mortimer~\cite{DixonMortimer1996}, \S\,2.6.}, giving abelian index $\alpha(\Sk)=k!/3^{k/3}$. Figure~\ref{fig:sk_irreps} illustrates the landscape.

\begin{figure}[!htbp]
\centering
\includegraphics[width=\columnwidth]{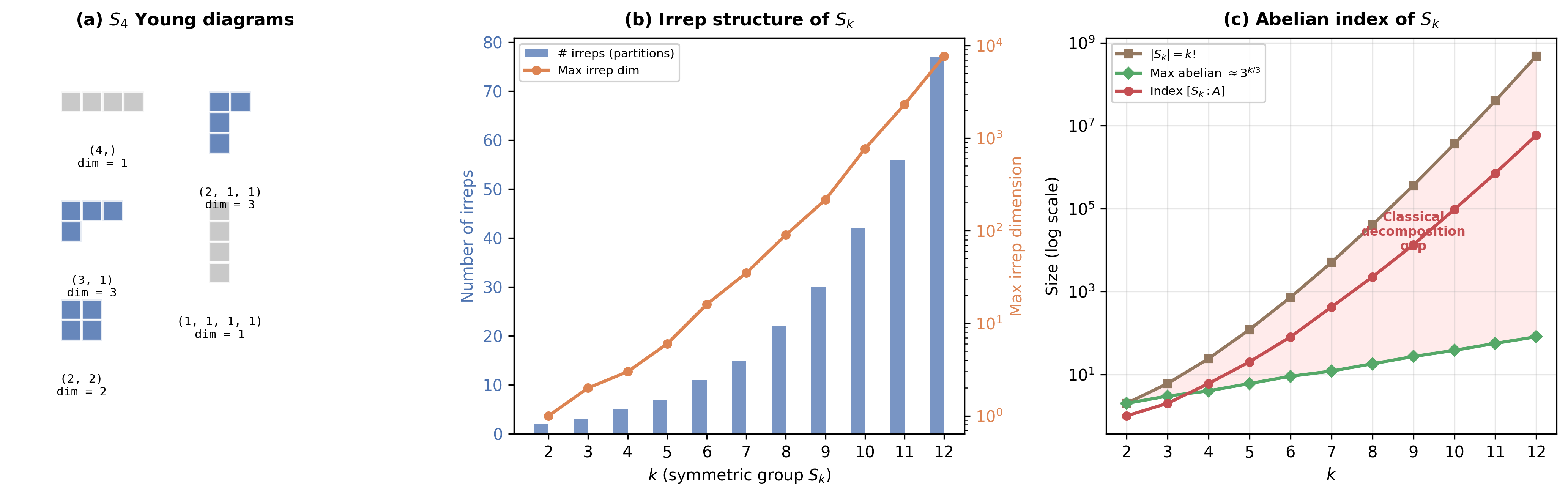}
\caption{Irreducible representation landscape for $\Sk$. (a)~Young diagrams and hook-length dimensions for $S_4$. (b)~Number of irreps and maximum dimension vs.~$k$. (c)~The super-exponential abelian index gap.}
\label{fig:sk_irreps}
\end{figure}

\section{On Extremal Conjugacy Classes}
\label{app:ecc_verification}

A natural route to establishing DMPC$\,\in\,$co-NP (Remark~\ref{rem:complexity_class}) would be to show that the minimiser of an $r$-sparse class function on~$\Sk$ always lies in a $\mathrm{poly}(r,k)$-sized ``extremal'' subset of conjugacy classes. We tested the candidate definition that extremal classes are those whose cycle type has at most~$r$ distinct part sizes. For each $(r,k)$ pair with $r\in\{2,3,4\}$ and $k\in\{3,\ldots,25\}$, we generated $10{,}000$ ($k\le15$) or $5{,}000$ ($k>15$) random $r$-sparse class functions, evaluated them on all $p(k)$ conjugacy classes, and recorded the minimiser.

\paragraph{Results.} The candidate definition holds trivially for $k\le 5$ (the extremal set covers all classes). For $k\ge 6$, counterexamples appear: at $(r{=}2,\,k{=}25)$, $47\%$ of trials have minimisers outside the extremal set, with cycle types exhibiting up to $5$ distinct part sizes. The maximum number of distinct parts in minimisers grows as $\Theta(\sqrt{k})$ \emph{independently of~$r$}. At threshold $\lceil\sqrt{k}\rceil$, the qualifying set covers ${\ge}94\%$ of all conjugacy classes across all tested~$k$, offering no useful restriction.

\paragraph{Conclusion.} No $\mathrm{poly}(r,k)$-sized characterisation of the minimising class is apparent from the data. The co-NP membership question for DMPC therefore remains open and may require fundamentally different techniques (e.g.\ algebraic certificates rather than combinatorial witnesses).

\section{Frustration Gap and Hybrid Approach}
\label{app:frustration}

\begin{proposition}[Frustration gap bound]
\label{prop:frustration_gap}
Let $G=(V,\mathcal{A})$ have cycle rank $\beta=m-n+1$, and let $\Delta_{\max} = \max_{(i,j)}\bigl(\max_{d}f_{ij}(d) - f_{ij}(d^*_{ij})\bigr)$. The tree-solver output $\mu_T$ satisfies $H(\mu_T) - H(\mu^*) \le \beta\cdot\Delta_{\max}$.
\end{proposition}

\begin{proof}
The tree solver achieves $\min_d f_{ij}(d)$ on all $n-1$ tree edges. Each of $\beta$ non-tree edges incurs at most $\Delta_{\max}$ excess cost. Therefore $H(\mu_T)\le\sum_{\mathrm{all}}f_{ij}(d^*_{ij})+\beta\cdot\Delta_{\max}\le H(\mu^*)+\beta\cdot\Delta_{\max}$.
\end{proof}

For small $\beta$, quantum Fourier recovery plus classical enumeration over $C^\beta$ residual configurations is exact when $\beta=O(\log n/\log C)$.

\paragraph{Code and data availability.}
Simulation scripts for numerical validation (Appendix~\ref{app:validation}), $p_{\min}$ bound verification, and ECC conjecture testing (Appendix~\ref{app:ecc_verification}) are available in the supplementary material accompanying this submission.

\bibliographystyle{quantum}
\bibliography{refs}

\begin{thebibliography}{10}

\bibitem{Hassanieh2012}
H.~Hassanieh, P.~Indyk, D.~Katabi, and E.~Price.
\newblock ``Simple and practical algorithm for sparse {F}ourier transform''.
\newblock \href{https://dx.doi.org/10.1137/1.9781611973099.93}{Proceedings of
  the 23rd Annual ACM--SIAM Symposium on Discrete Algorithms (SODA)Pages
  1183--1194}~(2012).

\bibitem{Gartner1975}
N.~H. Gartner, J.~D.~C. Little, and H.~Gabbay.
\newblock ``{MITROP}: A computer program for simultaneous optimisation of
  offsets, splits and cycle time''.
\newblock Traffic Engineering and Control {\bf 17}, 355--359~(1975).
\newblock  url:~\url{https://trid.trb.org/view/48074}.

\bibitem{GartnerDeshpande2009}
N.~H. Gartner and R.~M. Deshpande.
\newblock ``Harmonic analysis and optimization of traffic signal systems''.
\newblock In Transportation and Traffic Theory 2009: Golden Jubilee.
\newblock \href{https://dx.doi.org/10.1007/978-1-4419-0820-9_17}{Pages
  333--354}.
\newblock Springer~(2009).

\bibitem{DixitPech2026}
V.~V. Dixit and R.~Pech.
\newblock ``A constant time quantum algorithm for robust network signal
  coordination problem''~(2026).
\newblock  \href{http://arxiv.org/abs/2603.04758}{arXiv:2603.04758}.

\bibitem{Little1981}
J.~D.~C. Little, M.~D. Kelson, and N.~H. Gartner.
\newblock ``{MAXBAND}: A versatile program for setting signals on arteries and
  triangular networks''.
\newblock Transportation Research Record {\bf 795}, 40--46~(1981).
\newblock  url:~\url{https://trid.trb.org/view/175107}.

\bibitem{Gartner1991}
N.~H. Gartner, S.~F. Assmann, F.~Lasaga, and D.~L. Hou.
\newblock ``A multi-band approach to arterial traffic signal optimization''.
\newblock \href{https://dx.doi.org/10.1016/0191-2615(91)90013-9}{Transportation
  Research Part B {\bf 25}, 55--74}~(1991).

\bibitem{Serafini1989}
P.~Serafini and W.~Ukovich.
\newblock ``A mathematical model for periodic scheduling problems''.
\newblock \href{https://dx.doi.org/10.1137/0402049}{SIAM Journal on Discrete
  Mathematics {\bf 2}, 550--581}~(1989).

\bibitem{LindnerLiebchen2022}
N.~Lindner and C.~Liebchen.
\newblock ``Parameterized complexity of periodic timetabling''.
\newblock \href{https://dx.doi.org/10.1007/s10951-021-00713-z}{Journal of
  Scheduling {\bf 25}, 5--22}~(2022).

\bibitem{RamaswamiParhi1989}
V.~Ramaswami and K.~Parhi.
\newblock ``Distributed scheduling of broadcasts in a radio network''.
\newblock In IEEE INFOCOM.
\newblock \href{https://dx.doi.org/10.1109/INFCOM.1989.101537}{Pages 497--504}.
\newblock ~(1989).

\bibitem{Gronkvist2005}
J.~Gr{\"o}nkvist.
\newblock ``Interference-based scheduling in spatial reuse {TDMA}''.
\newblock PhD thesis.
\newblock Royal Institute of Technology (KTH), Stockholm.
\newblock ~(2005).
\newblock
  url:~\url{https://www.diva-portal.org/smash/get/diva2:12179/FULLTEXT01.pdf}.

\bibitem{Hale1980}
W.~K. Hale.
\newblock ``Frequency assignment: Theory and applications''.
\newblock \href{https://dx.doi.org/10.1109/PROC.1980.11899}{Proceedings of the
  IEEE {\bf 68}, 1497--1514}~(1980).

\bibitem{Aardal2007}
K.~I. Aardal, S.~P.~M. van Hoesel, A.~M. C.~A. Koster, C.~Mannino, and
  A.~Sassano.
\newblock ``Models and solution techniques for frequency assignment problems''.
\newblock \href{https://dx.doi.org/10.1007/s10479-007-0178-0}{Annals of
  Operations Research {\bf 153}, 79--129}~(2007).

\bibitem{BienstockVerma2019}
D.~Bienstock and A.~Verma.
\newblock ``Strong {NP}-hardness of {AC} power flows feasibility''.
\newblock \href{https://dx.doi.org/10.1016/j.orl.2019.03.008}{Operations
  Research Letters {\bf 47}, 215--218}~(2019).

\bibitem{LavaeiLow2012}
J.~Lavaei and S.~H. Low.
\newblock ``Zero duality gap in optimal power flow problem''.
\newblock \href{https://dx.doi.org/10.1109/TPWRS.2011.2160974}{IEEE
  Transactions on Power Systems {\bf 27}, 92--107}~(2012).

\bibitem{Kuramoto1975}
Y.~Kuramoto.
\newblock ``Self-entrainment of a population of coupled non-linear
  oscillators''.
\newblock In International Symposium on Mathematical Problems in Theoretical
  Physics.
\newblock \href{https://dx.doi.org/10.1007/BFb0013365}{Pages 420--422}.
\newblock ~(1975).

\bibitem{Strogatz2000}
S.~H. Strogatz.
\newblock ``From {K}uramoto to {C}rawford: exploring the onset of
  synchronization in populations of coupled oscillators''.
\newblock \href{https://dx.doi.org/10.1016/S0167-2789(00)00094-4}{Physica D:
  Nonlinear Phenomena {\bf 143}, 1--20}~(2000).

\bibitem{Awerbuch1985}
B.~Awerbuch.
\newblock ``Complexity of network synchronization''.
\newblock \href{https://dx.doi.org/10.1145/4221.4227}{Journal of the ACM {\bf
  32}, 804--823}~(1985).

\bibitem{OlfatiSaber2007}
R.~Olfati-Saber, J.~A. Fax, and R.~M. Murray.
\newblock ``Consensus and cooperation in networked multi-agent systems''.
\newblock \href{https://dx.doi.org/10.1109/JPROC.2006.887293}{Proceedings of
  the IEEE {\bf 95}, 215--233}~(2007).

\bibitem{Kempe2003}
D.~Kempe, J.~Kleinberg, and {\'E}.~Tardos.
\newblock ``Maximizing the spread of influence through a social network''.
\newblock In KDD.
\newblock \href{https://dx.doi.org/10.1145/956750.956769}{Pages 137--146}.
\newblock ~(2003).

\bibitem{Dauscha1985}
W.~Dauscha, H.~D. Modrow, and A.~Neumann.
\newblock ``On cyclic sequence types for constructing cyclic schedules''.
\newblock \href{https://dx.doi.org/10.1007/BF01919743}{Zeitschrift f{\"u}r
  Operations Research {\bf 29}, 1--30}~(1985).

\bibitem{Grover1996}
L.~K. Grover.
\newblock ``A fast quantum mechanical algorithm for database search''.
\newblock In Proceedings of the 28th Annual ACM Symposium on Theory of
  Computing (STOC).
\newblock \href{https://dx.doi.org/10.1145/237814.237866}{Pages 212--219}.
\newblock ~(1996).

\bibitem{Grover1997}
L.~K. Grover.
\newblock ``Quantum mechanics helps in searching for a needle in a haystack''.
\newblock \href{https://dx.doi.org/10.1103/PhysRevLett.79.325}{Physical Review
  Letters {\bf 79}, 325--328}~(1997).

\bibitem{Farhi2014}
E.~Farhi, J.~Goldstone, and S.~Gutmann.
\newblock ``A quantum approximate optimization algorithm''~(2014).
\newblock  \href{http://arxiv.org/abs/1411.4028}{arXiv:1411.4028}.

\bibitem{AlbashLidar2018}
T.~Albash and D.~A. Lidar.
\newblock ``Adiabatic quantum computation''.
\newblock \href{https://dx.doi.org/10.1103/RevModPhys.90.015002}{Reviews of
  Modern Physics {\bf 90}, 015002}~(2018).

\bibitem{Childs2003}
Andrew~M. Childs, Richard Cleve, Enrico Deotto, Edward Farhi, Sam Gutmann, and
  Daniel~A. Spielman.
\newblock ``Exponential algorithmic speedup by a quantum walk''.
\newblock \href{https://dx.doi.org/10.1145/780542.780552}{Proceedings of the
  35th ACM Symposium on Theory of Computing (STOC)Pages 59--68}~(2003).

\bibitem{HHL2009}
Aram~W. Harrow, Avinatan Hassidim, and Seth Lloyd.
\newblock ``Quantum algorithm for linear systems of equations''.
\newblock \href{https://dx.doi.org/10.1103/PhysRevLett.103.150502}{Physical
  Review Letters {\bf 103}, 150502}~(2009).

\bibitem{Garey1976}
M.~R. Garey, D.~S. Johnson, and L.~Stockmeyer.
\newblock ``Some simplified {NP}-complete graph problems''.
\newblock \href{https://dx.doi.org/10.1016/0304-3975(76)90059-1}{Theoretical
  Computer Science {\bf 1}, 237--267}~(1976).

\bibitem{Kendall1938}
M.~G. Kendall.
\newblock ``A new measure of rank correlation''.
\newblock \href{https://dx.doi.org/10.1093/biomet/30.1-2.81}{Biometrika {\bf
  30}, 81--93}~(1938).

\bibitem{Beals1997}
R.~Beals.
\newblock ``Quantum computation of {F}ourier transforms over symmetric
  groups''.
\newblock In Proceedings of the 29th Annual ACM Symposium on Theory of
  Computing (STOC).
\newblock \href{https://dx.doi.org/10.1145/258533.258548}{Pages 48--53}.
\newblock ~(1997).

\bibitem{NielsenChuang2000}
Michael~A. Nielsen and Isaac~L. Chuang.
\newblock ``Quantum computation and quantum information''.
\newblock \href{https://dx.doi.org/10.1017/CBO9780511976667}{Cambridge
  University Press}. ~(2000).

\bibitem{Bartholdi1989}
John~J. Bartholdi, III, Craig~A. Tovey, and Michael~A. Trick.
\newblock ``Voting schemes for which it can be difficult to tell who won the
  election''.
\newblock \href{https://dx.doi.org/10.1007/BF00303169}{Social Choice and
  Welfare {\bf 6}, 157--165}~(1989).

\bibitem{MooreRussellSchulman2005}
Cristopher Moore, Alexander Russell, and Leonard~J. Schulman.
\newblock ``The symmetric group defies strong {F}ourier sampling''.
\newblock In Proceedings of the 46th Annual IEEE Symposium on Foundations of
  Computer Science (FOCS).
\newblock \href{https://dx.doi.org/10.1109/SFCS.2005.73}{Pages 479--488}.
\newblock ~(2005).

\bibitem{KushilevitzMansour1993}
E.~Kushilevitz and Y.~Mansour.
\newblock ``Learning decision trees using the {F}ourier spectrum''.
\newblock In Proceedings of the 25th Annual ACM Symposium on Theory of
  Computing (STOC).
\newblock \href{https://dx.doi.org/10.1145/167088.167225}{Pages 455--464}.
\newblock ~(1993).

\bibitem{DixonMortimer1996}
John~D. Dixon and Brian Mortimer.
\newblock ``Permutation groups''.
\newblock \href{https://dx.doi.org/10.1007/978-1-4612-0731-3}{Volume 163 of
  Graduate Texts in Mathematics}.
\newblock Springer-Verlag. New York~(1996).

\bibitem{Barthelemy2011}
Marc Barth{\'e}lemy.
\newblock ``Spatial networks''.
\newblock \href{https://dx.doi.org/10.1016/j.physrep.2010.11.002}{Physics
  Reports {\bf 499}, 1--101}~(2011).

\bibitem{FrameRobinsonThrall1954}
J.~S. Frame, G.~de~B. Robinson, and R.~M. Thrall.
\newblock ``The hook graphs of the symmetric group''.
\newblock \href{https://dx.doi.org/10.4153/CJM-1954-030-1}{Canadian Journal of
  Mathematics {\bf 6}, 316--324}~(1954).

\end{thebibliography}

\end{document}